\documentclass[orivec,runningheads]{llncs}
\usepackage{color}
\usepackage{extarrows}
\usepackage{amsmath}
\usepackage{amssymb}
\usepackage{proof}
\usepackage{subfig}
\usepackage[all]{xy}
\usepackage{setspace}
\usepackage{url}
\usepackage{qtree}
\usepackage{enumerate}

\begin{document}
\title{A Fully Abstract Semantics for Value-passing CCS for Trees}
\author{Shichao Liu\inst{1,2}\and Thomas Ehrhard\inst{3}\and Ying Jiang\inst{1}}
\institute{State Key Laboratory of Computer Science, Institute of Software,\\
 Chinese Academy of Sciences, Beijing, China
\and University of Chinese Academy of Sciences,
Beijing, China \\ \texttt{\{liusc,jy\}@ios.ac.cn}
\and
CNRS, IRIF, UMR 8243, Univ Paris Diderot,\\
Sorbonne Paris Cit\'{e}, F-75205 Paris, France\\
\texttt{thomas.ehrhard@pps.univ-paris-diderot.fr}
}
\authorrunning{S. Liu, T. Ehrhard and Y. Jiang}
\titlerunning{A Fully Abstract Semantics for Value-passing CCS for Trees}
\maketitle

\begin{abstract}
This paper provides a fully abstract semantics for value-passing CCS for trees (VCCTS). The operational semantics is given both in terms of a reduction semantics and in terms of a labelled transition semantics. The labelled transition semantics is non-sequential, allowing more than one action occurring simultaneously. We develop the theory of behavioral equivalence by introducing both weak barbed congruence and weak bisimilarity. In particular, we show that weak barbed congruence coincides with weak bisimilarity on image-finite processes.
This is the first such result for a concurrent model with tree structures. Distributed systems can be naturally modeled by means of this graph-based system, and some examples are given to illustrate this.
\end{abstract}

\section{Introduction}\label{Intro}
Nowadays, reactive systems are widely used and are increasingly important in daily life.
Reactive systems continuously act and react in response to stimuli from their environments.
Process algebra, e.g. Hoare's CSP \cite{Hoare1985communicating}, Milner's CCS \cite{Milner1989} and the $\pi$-calculus \cite{milner1992picalculus}, is a branch of theoretical compute science that has been developed to apply to reactive systems. In the past three decades, we have seen a series of efforts in the field of process algebra to find suited primitives and tools for studying various properties of reactive systems. On the one hand, different extensions, by enriching the labelled transition systems (LTSs), have been proposed to describe concurrency and distribution, e.g. \cite{boudol1988non,boudol1994theory,castellani2001process,darondeau1990causal,degano1990partial,kiehn1994comparing}.
On the other hand, concurrency is usually obtained by mapping process algebras to models equipped with concepts of causality and concurrency, e.g. Petri Nets \cite{reisig1985petri} and Event Structures \cite{winskel1989introduction}. However, the latter approaches make it harder to use (or even prevent it from using) the standard tools, e.g. bisimilarity, of process algebras based on LTSs.

The widespread use of distributed system, e.g. web services and wireless networks, has raised new challenges to process algebras. Such systems are physically or logically distributed and highly dynamic because the connections between subcomponents of a system can vary. And the topology of connections can no longer be a simple one, e.g. in CCS all parallel subprocesses always connecting with each other and the communication is global. In wireless systems, the communication is local, i.e. a transmission only spans a limited area.
Graph-rewriting systems are important efforts for this, e.g. \cite{ferrari2006synchronised,konig2001observational}, however, they have the same problems as Petri Nets and Event Structures as explained above. This area is attracting the attention of many researchers, and process algebra is still in an intensely exploratory phase.

Recently, a new theory of CCS for trees (CCTS) \cite{Ehrhard2013ccts,Ehrhard2015dpc} has been proposed. One of the motivations of CCTS is to give a uniform extension of both CCS and top-down tree automata like CCS as a natural extension of finite automata with interactions.
However, more attention is paid to the definition of parallel composition in CCTS. In CCS, each channel (called symbol in this paper) is unary, while in CCTS a symbol is $n$-ary. In CCTS, a prefixed process $f\cdot(P_1,\ldots,P_n)$ can reduce to $n$ processes $P_i$ running concurrently without interactions between each other, by performing $f$. It can also communicate with another dually prefixed process $\overline{f}\cdot(Q_1,\ldots,Q_n)$, reducing to processes $P_i$ and $Q_i$, $i\in\{1,\ldots,n\}$.
To accommodate the recognition of top-down trees \cite{tata2007}, $P_i$ can communicate with $Q_i$, but cannot communicate with $Q_j$ if $i\neq j$, with $i,j\in\{1,\ldots,n\}$. $P_i$ can also communicate with other processes that can communicate with $f\cdot(P_1,\ldots,P_n)$ in a larger system.
And the case for $Q_j$ is similar.
To characterize communicating capacities, graphs are used to define parallel compositions. Subprocesses are located on the vertices of the graph, and two subprocesses can communicate through dual symbols if there exists an edge between them. Therefore, the topology of connections in systems can be captured by this graph-based model.

%

Process equivalence is a central idea in process algebras. Weak barbed congruence \cite{Milner1992barbed} is a natural way of saying that two processes have the same behaviours in all contexts. However, it is hard to prove congruence directly because one has to consider all possible contexts. Instead, more tractable techniques have been used to establish it and a typical tool is {\it bisimilarity}. In CCTS, the authors defined both internal reductions (used to derive weak barbed congruence) and an LTS (used to derive weak bisimilarity). They proved that weak bisimilarity in CCTS implied weak barbed congruence, i.e. {\it soundness}. 
However, the converse direction, i.e. {\it completeness}, was not established.

In \cite{liu2016vccts}, the authors extended the syntax of CCTS to value-passing CCTS (VCCTS), whose symbols can receive and send data values. Just like CCS, adding explicit value passing to CCTS does not increase the expressiveness but improves readability. VCCTS \cite{liu2016vccts} has the same communication constraints as CCTS. The operational semantics is non-sequential and allows more than one action happening in the same transition labelled by multisets, while in CCTS only one action can occur in each transition. As in CCTS, only soundness was proved in \cite{liu2016vccts}, and the completeness was not established either. Soundness alone does not tell us whether bisimilarity are applicable to many processes. For instance, the identity relation is included in weak barbed congruence and it is sound. But it does not show any interesting proof techniques.

In this paper, we propose a fully abstract semantics for VCCTS focusing on canonical processes (cf. Section \ref{canonical_pro}) like CCTS. To obtain the completeness of weak bisimilarity, we provide a less restrictive semantics for VCCTS.
Compared to CCTS, the main difference is that we relax the constraints of the communicating capacities after communications and this is the key to prove completeness. In this paper, an input process $f(x)\cdot(P_1,\ldots,P_n)$ and an output process $\overline{f}(v)\cdot(Q_1,\ldots,Q_n)$ can communicate and reduce to $P_i\{v/x\}$ (obtained by substituting $v$ for $x$ in $P_i$) and $Q_i$, where we use $x$ for data variable and $v$ for data value, and $i\in\{1,\ldots,n\}$.
Different from CCTS, $P_i\{v/x\}$ can communicate with each $Q_j$ for $i,j\in\{1,\ldots,n\}$.
The communicating capacities for $Q_j$ are similar (cf. Section \ref{Semantics}). Therefore, the semantics in this paper is different from the ones in CCTS and VCCTS \cite{liu2016vccts}.
The operational semantics is given both in terms of a reduction semantics and in terms of a labelled transition semantics.
And the latter is non-sequential.
We also introduce weak barbed congruence and weak bisimilarity, and prove the main result of this paper that the two relations coincide on image-finite systems.

For completeness, we need to prove that for all image-finite canonical processes $P$ and $Q$, if $P$ and $Q$ are not weakly bisimilar, then $P$ and $Q$ are not weakly barbed congruent.
The proof usually needs to define a stratification of weak bisimilarity used by induction on number of steps of the weak bisimilarity \cite{Sangiorgi:2011}. 
There are examples in CCTS to show that there are canonical processes $P$ and $Q$ being not weakly bisimilar and they are not weakly barbed congruent either.
However, the completeness of CCTS is very hard (if not impossible) to prove using the method in Theorem \ref{completeness} (cf. Section \ref{Soundness}), because of the communication constraints of CCTS.

This new variation of CCTS and VCCTS \cite{liu2016vccts} retains the advantages of the original CCTS, e.g. tree structures and local connections, etc.
We show that if a tree can be recognized by a top-down tree automaton then the processes corresponding to the tree and the automaton, running concurrently, can reduce to an idle process (Section \ref{top_down_automata}). But the converse direction, which was valid in CCTS, is no longer valid.
Sangiorgi was the first to prove that barbed congruence and bisimilarity coincide in a weak version for both CCS and the $\pi$-calculus in his PhD thesis \cite{sangiorgi1993expressing}.
We embed value-passing CCS (VCCS) \cite{Milner1989} in VCCTS. Thus, a location version of VCCS is obtained similar to \cite{boudol1994theory,castellani2001process} for location versions of CCS, and a fully abstract non-sequential semantics is given to VCCS.
Some examples are also investigated in VCCTS.


The rest of this paper is organized as follows. In Section \ref{Syntax}, we introduce the syntax of VCCTS. In Section \ref{Semantics}, we develop two kinds of operational semantics for VCCTS, and weak barbed congruence and weak bisimilarity are defined. In Section \ref{Soundness}, we prove that the two relations coincide. Section \ref{Discussion} investigates top-down tree automata and VCCS in VCCTS. In Section \ref{Exmaples}, we apply VCCTS to the Alternating Bit Protocol. Related work and conclusions are discussed in Section \ref{Related_work} and Section \ref{Conclusions}, respectively.

For lack of space, all the proofs are omitted, but they can be found in Appendices.
\section{Syntax of Value-passing CCTS}\label{Syntax}
\paragraph{Expression.}
We assume that ${\bf Var}$ is a set of data variables ranged over by $x,y$, etc., and that ${\bf Val}$ is a set of data values ranged over by $v,v_1$, etc. Suppose that ${\bf Exp}$ is a set of arithmetic expressions and at least includes ${\bf Var}$ and ${\bf Val}$, ranged over by $e,e_1$, etc. We also assume that ${\bf BExp}$ is a set of boolean expressions and includes $\{\it false,true\}$, ranged over by $b,b_1$, etc.
Let ${\sf fv}(e)$ and ${\sf fv}(b)$ be the sets of (free) data variables in $e$ and $b$, respectively.  If ${\sf fv}(e) = \emptyset$, we say that $e$ is closed, and similarly for $b$.
The substitution of a value $v$ for a variable $x$ is defined as usual, denoted by $e\{v/x\}$ and $b\{v/x\}$. We use ${\vec x}$ and ${\vec v}$ to range over vectors of variables and values, respectively.

\paragraph{Symbols.}
Let ${\bf N}$ be the set of natural numbers. Let $\Sigma = (\Sigma_n)_{n\in {\bf N}}$ be a signature, i.e. a set of ranked symbols. For each symbol $f\in\Sigma_n(n\ge 1)$, there is a co-symbol $\overline{f}$ such that $\overline{f}\neq f$.
Particularly, we assume that there is only one symbol of arity $0$, denoted by $\ast$. Let $\overline{\ast}= \ast$ for simplicity.
Let $\overline{\Sigma}_n = \Sigma_n\cup \{\overline{f}\mid f\in\Sigma_n\}$, $\overline{\Sigma}=(\overline{\Sigma}_n)_{n\in{\bf N}}$, and $\overline{\overline{f}}=f$ by convention.
We use $I,J$ to stand for subsets of $\overline{\Sigma}$, and use $\overline{I}$ for the set of complements of elements in $I$.

\paragraph{Graphs.}
Let {\sf Loc} be a countable set of locations ranged over by $p,q$, etc.
A finite graph $G=(|G|,\frown_G)$ consists of a finite set of locations $|G|\subseteq {\sf Loc}$ and a set of edges $\frown_G$ which is a symmetric and irreflexive binary relation on $|G|$.
Let $E$ and $F$ be disjoint sets of locations with $p\in E$, and we define $E[F/p]=(E\setminus \{p\})\cup F$.
Let $G$ and $H$ be graphs with $|G|\cap |H|=\emptyset$ and let $p\in |G|$.
We define a graph $G[H/p]= (|G[H/p]|,\frown_{G[H/p]})$ with $|G[H/p]|=|G|[|H|/p]$ and $q\frown_{G[H/p]}r$ if $q\frown_G r$, or $q\frown_H r$, or $q\frown_G p$ and $r\in|H|$, or $r\frown_G p$ and $q\in |H|$.

Let $\mathcal{V}$ be a countable set of process variables ranged over by $X,Y$, etc. $\mathcal{K}$ is a set of constants, ranged over by $A,B$, etc. For each $A\in \mathcal{K}$, we assume that there is an assigned arity, a non-negative integer representing the number of parameters that $A$ takes.
${\bf Pr}$ is the set of processes, defined as follows:
\begin{center}
\begin{tabular}{crl}
  $P,Q$ & $::=$ & $\ast\mid{\bf 0}\mid X\mid f(x)\cdot (P_1,\ldots,P_n)\mid\overline{f}(e)\cdot (P_1,\ldots,P_n)\mid G\langle\Phi\rangle  $ \\
   &  & $ \mid P+Q\mid P\backslash I\mid {\bf if}~b~{\bf then}~P_1~{\bf else}~P_2\mid A({\vec v})$ \\
\end{tabular}
\end{center}
where $X\in \mathcal{V}$, $x\in{\bf Var}$, $e\in{\bf Exp}$, $b\in {\bf BExp}$, $f\in \Sigma_n (n\ge 1)$, $P_1,\ldots,P_n \in {\bf Pr}$, $G$ is a finite graph, $\Phi$ is a function from $|G|$ to ${\bf Pr}$, $I$ is a finite subset of $\Sigma$ and ${\vec v}$ is consistent with the assigned arity of $A$.

The symbol $\ast$ does not pass any value, and
$\ast()\cdot()$ is simply written $\ast$ if there is no confusion.
$\ast$ is an idle process different from the empty sum ${\bf 0}$.
Sum operator $+$ and symbol restriction $\backslash$ have ordinary meanings, and (co-)symbols mentioned in $I\cup\overline{I}$ are bound in $P\backslash I$.
Data variable $x$ is bound in the input process $f(x)\cdot (P_1,\ldots,P_n)$;
data variables appearing in $e$ are free in the output process $\overline{g}(e)\cdot(Q_1,\ldots,Q_m)$.
$G\langle\Phi\rangle$ is the parallel composition of the processes $\Phi(p)$ with $p\in |G|$, and for $p,q\in|G|$, $\Phi(p)$ and $\Phi(q)$ can communicate through dual symbols if $p\frown_G q$. Processes $\Phi(p)$ are called components of $G\langle\Phi\rangle$. A process ${\bf if}~b~{\bf then}~P~{\bf else}~Q$ behaves as $P$ if the value of $b$ is $true$, and as $Q$ otherwise. $A({\vec v})$ denotes a process defined by a possibly recursive definition of the form $A({\vec x})\stackrel{\rm def}{=}P$. Process variable $X$ will be used in the definition of context later.

For substitutions, let ${\sf fv}(P)$ represent the free data variables in $P$, and $P$ is {\it data closed} if ${\sf fv}(P)=\emptyset$.
$P\{v/x\}$ is the result of substituting $v$ for every free occurrence of $x$ in $P$, and similarly for $P\{{\vec v}/{\vec x}\}$.
$Q[P/X]$ represents the substitution of $P$ for every free occurrence of $X$ in $Q$. 
Given $P=G\langle\Phi\rangle$ and $Q=H\langle\Psi\rangle$ with $p\in |G|$ and $|G|\cap |H|=\emptyset$, $P[Q/p]$ represents the process $G[H/p]\langle\Phi^{\prime}\rangle$ with $\Phi^{\prime}(p^{\prime}) = \Phi(p^{\prime})$ for $p^{\prime}\notin |H|$ and $\Phi^{\prime}(p^{\prime}) = \Psi(p^{\prime})$ for $p^{\prime}\in |H|$. In general, a substitution may require $\alpha$-conversions on data variables and symbols.

Given graphs $G$ and $H$ with $|G|\cap |H| = \emptyset$ and $D\subseteq |G|\times|H|$, we define a new graph $K= G \oplus_D H$ with $|K|=|G|\cup|H|$ and $p\frown_K q$ if $p\frown_G q$ or $p\frown_H q$ or $(p,q)\in D$ for any $p,q\in|K|$. When $D=\emptyset$, we let $G\oplus H = G\oplus_D H$.
Given $P=G\langle\Phi\rangle\backslash I$ and $Q=H\langle\Psi\rangle\backslash J$ with $|G|\cap |H| = \emptyset$, $I\cap J = \emptyset$ (always possible with $\alpha$-conversions on (co-)symbols), and  $D\subseteq |G|\times|H|$, we define process $P\oplus_D Q$ as
$(G\oplus_D H)\langle\Phi\cup\Psi\rangle\backslash (I\cup J)$. When $D=|G|\times|H|$, $P\oplus_D Q$ is written as $P\mid Q$ for simplicity. We write $P\oplus Q$ for $P\oplus_{\emptyset} Q$.
More generally, $\oplus{\vec P}$ stands for  $P_1\oplus\cdots\oplus P_n$ when ${\vec P} = (P_1,\ldots,P_n)$. When we consider $P_1,\ldots,P_n$ together, we always assume that their associated graphs are pairwise disjoint.

\subsection{Canonical Processes}\label{canonical_pro}
Roughly speaking, a process is canonical if all sums in it are guarded.
Formally, we define {\it canonical processes} ({\sf CP}), {\it canonical guarded sums} ({\sf CGS}) and {\it recursive canonical guarded sums} ({\sf RCGS}) following  \cite{Ehrhard2013ccts} by mutual induction as follows:
\begin{center}\small
$\infer[]{X\in {\sf CP}}{X\in \mathcal{V}}$ $\infer[]{A({\vec v})\in {\sf CP}}{P\in {\sf CP}~A({\vec x})\stackrel{{\rm def}}{=}P}$ $\infer[]{G\langle\Phi\rangle\in{\sf CP}}{\Phi: |G|\rightarrow {\sf RCGS}}$ $\infer[]{P\backslash I\in {\sf CP}}{P\in {\sf CP}\quad I\subseteq \Sigma}$ $\infer[]{{\bf 0}\in {\sf CGS}}{}$ $\infer[]{\ast\in {\sf CGS}}{}$ \\
$\infer[]{f(x)\cdot(P_1,\ldots,P_n)\in {\sf CGS}}{x\in {\bf Var}\quad f\in\Sigma_n\quad P_1,\ldots,P_n\in{\sf CP}}$\qquad
$\infer[]{\overline{f}(e)\cdot(Q_1,\ldots,Q_n)\in {\sf CGS}}{e\in {\bf Exp}\quad f\in\Sigma_n\quad Q_1,\ldots,Q_n\in{\sf CP}}$\\
$\infer[]{S_1+S_2\in{\sf CGS}}{S_1,S_2\in{\sf CGS}}$ $\infer[]{{\bf if}~ b ~{\bf then}~ S_1~ {\bf else}~ S_2\in{\sf CGS}}{b\in{\bf BExp}~~ S_1,S_2\in{\sf CGS}}$
$\infer[]{S\in {\sf RCGS}}{S\in {\sf CGS}}$ $\infer[]{A({\vec v}) \in {\sf RCGS}}{S\in {\sf RCGS} ~~ A({\vec x})\stackrel{{\rm def}}{=}S}$
\end{center}
If $P= G\langle\Phi\rangle\backslash I$ (or $P= G\langle\Phi\rangle$) is a canonical process, we denote $|P| = |G|$. Meanwhile, for $p\in |G|$, let $P(p)$ represent $\Phi(p)$ and let $\frown_P$ represent $\frown_G$. A recursive process can also be built by $A({\vec x})\stackrel{{\rm def}}{=}P[A({\vec v})/X]$, e.g. $A_1\stackrel{\rm def}{=}f(x)\cdot(X)[A_1/X] = f(x)\cdot(A_1)$.
\begin{lemma}\label{lemma:cs}
If $R$ and $P$ are canonical processes, then $R[P/X]$ is a canonical process.
If $R$ is a (recursive) canonical guarded sum, so is $R[P/X]$.
\end{lemma}

{\it In the rest of this paper, we only consider data-closed canonical processes, and let ${\sf Proc}$ represent this set.}

We assume the existence of an evaluation function ${\sf eval}$ for the closed expressions in {\bf Exp} and {\bf BExp}.
For each recursive canonical guarded sum $S$, there is a canonical guarded sum ${\sf cs}(S)$ defined by:\\
$${\sf cs}(S)=
  \left\{
    \begin{array}{ll}
      S & \hbox{if } S = f(x)\cdot(P_1,\ldots,P_n) + S_3, \hbox{ or } S = {\bf 0}, \hbox{ or } S = \ast,\\
        & \hbox{\quad or }  S = \overline{f}(e)\cdot(P_1,\ldots,P_n) + S_3;\\
      {\sf cs}(S_1) & \hbox{if } S=({\bf if}~b~{\bf then}~S_1~{\bf else}~S_2) + S_3 \hbox{ with } {\sf eval}(b) = {\it true};\\
      {\sf cs}(S_2) & \hbox{if } S=({\bf if}~b~{\bf then}~S_1~{\bf else}~S_2) + S_3 \hbox{ with } {\sf eval}(b) = {\it false};\\
      {\sf cs}(T\{{\vec v}/{\vec x}\}) & \hbox{if } S=A({\vec v}) \hbox{ and } A({\vec x})\stackrel{\rm def}{=}T.\\
    \end{array}
  \right.
$$
From the rule if $S_1,S_2\in{\sf CGS}$ then $S_1+S_2 \in{\sf CGS}$, it is obvious that the terms in summations are always guarded canonical sums.
Because we can reorder the terms in a summation (commutativity), we just write a prefixed process on the left side of $+$.
From the definition of ${\sf RCGS}$, one can easily see that the function is well defined.

\section{Operational Semantics}\label{Semantics}
\subsection{Internal Reduction}

The internal reduction over ${\sf Proc}$ is of the form $P\xrightarrow[]{}P^{\prime}$, defined in Fig. \ref{Internal_Reduction}.
In (R-React), for $p,q\in |P|$, $p\frown_P q$ means that processes $P(p)$ and $P(q)$ can interact on dual symbols, where
\begin{itemize}
  \item ${\sf cs}(P(p))$ is of the form $f(x)\cdot(P_1,\ldots,P_n) + S$;
  \item ${\sf cs}(P(q))$ is of the form $\overline{f}(e)\cdot(Q_1,\ldots,Q_n)+T$ and ${\sf eval}(e)=v$.
\end{itemize}
Then, $P$ can reduce to $P^{\prime}$, where $P^{\prime}$ is defined as follows: first, $|P^{\prime}|=(|P|\setminus \{p,q\})\cup \bigcup_{i=1}^n|P_i\{v/x\}|\cup \bigcup_{i=1}^n|Q_i|$; then, $\frown_{P^{\prime}}$ is the least symmetric relation on $|P^{\prime}|$ such that, for any $p^{\prime},q^{\prime}\in|P^{\prime}|$, $p^{\prime}\frown_{P^{\prime}} q^{\prime}$ if one of the following cases is satisfied:
        \begin{itemize}
          \item $p^{\prime}\frown_{P_i\{v/x\}} q^{\prime}$ or
                 $p^{\prime}\frown_{Q_i} q^{\prime}$ for some $i=1,\ldots,n$ \hfill $(a)$
          \item $p^{\prime}\in \bigcup_{i=1}^n|P_i\{v/x\}|$ and
                $q^{\prime}\in \bigcup_{i=1}^n|Q_i|$ \hfill $(b)$
          \item $\{p^{\prime},q^{\prime}\}\nsubseteq \bigcup_{i=1}^n|P_i\{v/x\}|
                 \cup\bigcup_{i=1}^n|Q_i|$ and $\lambda(p^{\prime}) \frown_P
                 \lambda(q^{\prime})$ \hfill $(c)$
        \end{itemize}
where $\lambda:|P^{\prime}|\rightarrow |P|$ is a {\it residual function}
        satisfying $\lambda(p^{\prime})=p$ if $p^{\prime}\in \bigcup_{i=1}^n|P_i\{v/x\}|$, $\lambda(p^{\prime})=q$ if $p^{\prime}\in \bigcup_{i=1}^n|Q_i|$ and $\lambda(p^{\prime})=p^{\prime}$ otherwise; at last, $P^{\prime}(p^{\prime})=P_i\{v/x\}(p^{\prime})$ if $p^{\prime}\in|P_i\{v/x\}|$, $P^{\prime}(p^{\prime})=Q_i(p^{\prime})$ if $p^{\prime}\in |Q_i|$
        and $P^{\prime}(p^{\prime})= P(p^{\prime})$ if $p^{\prime} \notin \bigcup_{i=1}^n|P_i\{v/x\}|\cup \bigcup_{i=1}^n|Q_i|$, with $i\in \{1,\ldots,n\}$.

The connection between $p$ and $q$ in $P$ is inherited by the vertices in $|P_i\{v/x\}|$ and $|Q_j|$ (cf. (b)).
The connections between $p$ and other vertices of $P$, distinct from $q$, are
inherited by the vertices in $|P_i\{v/x\}|$, and similarly for $q$ and $|Q_i|$ (cf. (c)). Rules (R-Res) and (R-Con) model restriction and constant cases as usual, respectively.
Let $\xrightarrow[]{}^\ast$ denote the reflexive and transitive closure of $\xrightarrow[]{}$.
\begin{figure}[!tp]\small
$$\infer[\mbox{(R-React)}]{P\xrightarrow[]{}P^{\prime}}{p,q\in|P|~~P(q),P(p)\in {\sf RCGS}~~p\frown_P q}$$
$$\infer[\mbox{(R-Res)}]{P\backslash I \xrightarrow[]{} P^{\prime}\backslash I}{P\xrightarrow[]{} P^{\prime}}\quad
\infer[\mbox{(R-Con)}]{A({\vec v})\xrightarrow[]{}P^{\prime}}{P\{{\vec v}/{\vec x}\}\xrightarrow[]{}P^{\prime}~~A({\vec x})\stackrel{{\rm def}}{=}P}$$
  \caption{Internal Reductions}\label{Internal_Reduction}
\end{figure}
\subsection{Weak Barbed Congruence}
To endow VCCTS with a non-sequential semantics, it could express more than one action occurring simultaneously.
Following Milner and Sangiorgi \cite{Milner1992barbed}, we use {\it barb} to describe the observable information.
\begin{definition}[\bf Barb for Recursive Canonical Guarded Sum]
Let $f\in \overline{\Sigma}$ and $P\in{\sf RCGS}$. We say that $f$ is a barb of $P$, written $P\downarrow_{f}$, if one of the following holds:
\begin{itemize}
  \item $f = g$, $g\in \Sigma$ and ${\sf cs}(P)$ is of the form $g(x)\cdot(P_1,\ldots,P_n) + S$;
  \item $f = \overline{g}$, $g\in \Sigma$ and ${\sf cs}(P)$ is of the form $\overline{g}(e)\cdot(P_1,\ldots,P_n) + S$.
\end{itemize}
\end{definition}
\begin{definition}[\bf Barb for Canonical Process]
We say that a finite subset $B$ of $\overline{\Sigma}$ is a barb of a canonical process $Q$, written $Q\downarrow_B$, if there exist distinct locations $q_i \in |Q|$, such that $Q(q_i)\downarrow_{f_i}$ for each $f_i\in B$ and, moreover, $f_i\notin I$ and $\overline{f_i}\notin I$ if $Q$ is of the form $P\backslash I$.
\end{definition}

\begin{example}[Barbs]
Let $f,g\in\Sigma_2$ with $f\neq g$ and process $P = \overline{f}(3)\cdot(\ast,\ast)\mid \overline{g}(4)\cdot(\ast, \ast)$. The associated graph of $P$ is a complete graph with two vertices, $\{1,2\}$, and ${\sf cs} (P(1)) = \overline{f}(3)\cdot(\ast,\ast)$ and ${\sf cs}(P(2)) = \overline{g}(4)\cdot(\ast,\ast)$. We have $P(1)\downarrow_{\overline{f}}$ and $P(2)\downarrow_{\overline{g}}$. Thus, $P\downarrow_{\{\overline {f}\}}$, $P\downarrow_{\{\overline{g}\}}$ and $P\downarrow_{\{\overline{f},\overline{g}\}}$.
\qed\end{example}

\begin{definition}[\bf Weak Barbed Bisimulation]
A binary relation $\mathcal{B}$ on ${\sf Proc}$ is a weak barbed bisimulation if it is symmetric and whenever $(P,Q)\in \mathcal{B}$ the following conditions are satisfied:
\begin{itemize}
  \item for any $P^{\prime} \in {\sf Proc}$, if $P\xrightarrow[]{}^\ast P^{\prime}$, then there exists $Q^{\prime}$ such that $Q\xrightarrow[]{}^{\ast}Q^{\prime}$ and $(P^{\prime},Q^{\prime})\in\mathcal{B}$;
  \item for any $P^{\prime} \in {\sf Proc}$ and any finite set $B\subseteq\overline{\Sigma}$, if $P\xrightarrow[]{}^{\ast}P^{\prime}$ and $P^{\prime}\downarrow_{B}$, then there exists $Q^{\prime}\in {\sf Proc}$ such that $Q\xrightarrow[]{}^{\ast}Q^{\prime}$ and   $Q^{\prime}\downarrow_{B}$.
\end{itemize}
Weak barbed bisimilarity, written $\stackrel{\bullet}{\approx}$, is the union of all weak barbed bisimulations.
\end{definition}

\begin{lemma}\label{lemma:barbbisim}
$\stackrel{\bullet}{\approx}$ is an equivalence relation.
\end{lemma}

We intend to investigate an important relation in weak barbed bisimilarity, i.e. weak barbed congruence with respect to one-hole contexts.
Given a process variable $Y$, a $Y${\it -context} is a canonical process $R$ containing only one free occurrence of $Y$, and $Y$ does not occur in any subprocess of $R$ of a recursive form $A({\vec x})\stackrel{\rm def}{=}R^{\prime}[A({\vec v})/X]$. A relation $\mathcal{R}\subseteq {\sf Proc}\times{\sf Proc}$ is a congruence if it is an equivalence and for any $Y${\it -context} $R$, $(P, Q)\in \mathcal{R}$ implies $(R[P/Y],R[Q/Y])\in \mathcal{R}$.
\begin{proposition}\label{prop:congruence}
For any equivalence $\mathcal{R}\subseteq {\sf Proc}\times{\sf Proc}$, there exists a largest congruence $\overline{\mathcal{R}}$ contained in $\mathcal{R}$. This relation is characterized by $(P,Q)\in \overline{\mathcal{R}}$ if and only if for any $Y$-context $R$ one has
$(R[P/Y],R[Q/Y])\in \mathcal{R}$.
\end{proposition}

\begin{definition}
Processes $P$ and $Q$ are weakly barbed congruent if $R[P/Y]\stackrel{\bullet}{\approx}R[Q/Y]$ for every $Y$-context $R$, denoted by $P \cong Q$.
\end{definition}

$\cong$ is the largest congruence included in $\stackrel{\bullet}{\approx}$ by Proposition \ref{prop:congruence}.

\subsection{Localized Labelled Transition Systems}
There are {\it early} semantics and {\it late} semantics, according to the time when the receiving of a value takes place in an input transition. In this paper, we adopt early semantics. 
To reflect the concurrent/distributed information in syntax, we add information of locations to transitions, obtaining {\it localized labelled transition systems} (LLTSs) over ${\sf Proc}$. In an LLTS, unrelated actions (see Definition \ref{Def_unrelated})
are allowed to happen in the same transition.

Let $Act=\{fv,\overline{f}v\mid v\in {\bf Val}, f\in \Sigma\}$ be the set of actions, ranged over by $\alpha$, $\alpha_1$, etc.
Given $\alpha = fv$, its dual action is $\overline{\alpha}=\overline{f}v$, and similarly for $\alpha = \overline{f}v$.
We define a function ${\sf symb}: Act \rightarrow \overline{\Sigma}$, and ${\sf symb}(fv)= f$ for $fv\in Act$.
Single-labelled transitions, defined in Fig. \ref{LTS1}, are of the form $P\xrightarrow[\lambda]{\delta}P^{\prime}$, where $\lambda$ is a residual function keeping the traces of locations during transitions and the label $\delta$ is defined as
$$\delta::=p:fv\cdot({\vec L})\mid p:\overline{f}v\cdot({\vec L})\mid \tau$$
where $p\in{\sf Loc}$, $f\in\Sigma$ and ${\vec L}=(L_1,\ldots,L_n)$ is a vector of the sets of locations. 

\begin{example}[Local Connections]
A system $S$ consists of a transmitter $A_1$ and two receivers $A_2$ and $A_3$, only $A_2$ connecting with $A_1$. The structure and connections of $S$ can be naturally expressed in VCCTS as follows:
$A_1 \stackrel{\rm def}{=} \overline{f}(5)\cdot(A_1), A_2\stackrel{\rm def}{=} f(x)\cdot(A_2), A_3 \stackrel{\rm def}{=} f(x)\cdot(A_3)$ and $S\stackrel{\rm def}{=}G\langle\Phi\rangle$,
where $f\in\Sigma_1$, $G=(\{1,2,3\},\{(1,2)\})$, $\Phi(i) = A_i$ for $i\in\{1,2,3\}$. That is $S\stackrel{{\rm def}}{=}(A_1\mid A_2)\oplus A_3$, and $S$ evolves as
$$ \infer[]{(A_1\mid A_2)\oplus A_3\xrightarrow[\mathrm{Id}]{\tau}(A_1\mid A_2) \oplus A_3}{A_1\xrightarrow[\mathrm{Id}]{1:\overline{f}5:(\{1\})}A_1\qquad  A_2\xrightarrow[\mathrm{Id}]{2:f5:(\{2\})}A_2}$$
$A_3$ cannot communicate with $A_1$, since there is no connection between them.
\qed\end{example}

Given a multiset $\Delta$ of labels, we use $\Delta(\delta)$ to represent the number of occurrences of $\delta$ in $\Delta$. 
We define ${\sf size}(\Delta) = \sum_{\delta\in \Delta}\Delta(\delta)$ to figure out the size of multiset $\Delta$.
We use $\Delta_n^{\tau}$ to represent a multiset which only contains $n$ $\tau$s, i.e. ${\sf size}(\Delta_n^{\tau})=\Delta_n^{\tau}(\tau)=n$. The union $\uplus$ and the difference $\backslash\!\!\backslash$ of multisets satisfy: $(\Delta_1\uplus\Delta_2)(\delta) = \Delta_1(\delta)+\Delta(\delta)$ and $(\Delta_1\backslash\!\!\backslash\Delta_2)(\delta)=
\mbox{max}(0,\Delta_1(\delta)-\Delta_2(\delta))$.

\begin{definition}[\bf Unrelated Action \cite{liu2016vccts}]\label{Def_unrelated}
Actions $\alpha_1$ and $\alpha_2$ are unrelated if ${\sf symb}(\alpha_1)$ $\neq {\sf symb}(\alpha_2)$. A multiset of labels $\Delta$ is pairwise unrelated, denoted by ${\sf PUnrel}(\Delta)$, if for every $(p:\alpha_1\cdot({\vec L_1}), q:\alpha_2\cdot({\vec L_2}))\in \Delta\times \Delta$ with $p\neq q$, $\alpha_1$ and $\alpha_2$ are unrelated.
\end{definition}

\begin{figure}[!tp]\small
$$\infer[(\mbox{Input})]{P\xrightarrow[\lambda]{p:fv\cdot({\vec L})}P^{\prime}}{p\in|P|\quad P(p)\in {\sf RCGS}}$$
 \begin{itemize}
        \item ${\sf cs}(P(p))$ is of the form $f(x)\cdot(P_1,\ldots,P_n) + S$;
        \item ${\vec L} = (|P_1\{v/x\}|, \ldots, |P_n\{v/x\}|)$;
        \item $P^{\prime} = P[\oplus {\vec P}/p]$ with ${\vec P} = (P_1\{v/x\},\ldots,P_n\{v/x\})$;
        \item $\lambda:|P^{\prime}|\rightarrow |P|$ is defined by $\lambda(p^{\prime}) = p$ if $p^{\prime}\in \bigcup_{i=1}^n L_i$ and $\lambda(p^{\prime}) = p^{\prime}$ otherwise.
      \end{itemize}
$$\infer[(\mbox{Output})]{P\xrightarrow[\lambda]{p:\overline{f}v\cdot({\vec L})}P^{\prime}}{p\in|P|\quad P(p)\in {\sf RCGS}}$$
\begin{itemize}
  \item  ${\sf cs}(P(p))$ is of the form $\overline{f}(e)\cdot(P_1,\ldots,P_n) + S$ and ${\sf eval}(e)=v$;
  \item  ${\vec L}= (|P_1|, \ldots, |P_n|)$, and $P^{\prime} = P[\oplus {\vec P}/p]$ with ${\vec P} = (P_1,\ldots,P_n)$;
  \item $\lambda:|P^{\prime}|\rightarrow |P|$ is defined by $\lambda(p^{\prime}) = p$ if $p^{\prime}\in \bigcup_{i=1}^n L_i$ and $\lambda(p^{\prime}) = p^{\prime}$ otherwise.
\end{itemize}
$$\infer[(\mbox{Com1})]{P\oplus_D Q \xrightarrow[\lambda]{\tau} P^{\prime}\oplus_{D^{\prime}}Q^{\prime}}{P\xrightarrow[\lambda_1]{p:\alpha\cdot({\vec L})}P^{\prime}\quad Q\xrightarrow[\lambda_2]{q:\overline{\alpha}\cdot({\vec H})}Q^{\prime}\quad (p,q)\in D  }$$
\begin{itemize}
  \item  $\lambda:|P^{\prime}|\cup |Q^{\prime}|\rightarrow |P|\cup |Q|$ is defined by $\lambda(p^{\prime})=\lambda_1(p^{\prime})$ if $p^{\prime}\in |P^{\prime}|$ and $\lambda(q^{\prime})=\lambda_2(q^{\prime})$ if $q^{\prime}\in |Q^{\prime}|$;
  \item  $(p^{\prime},q^{\prime})\in D^{\prime}$ if either $p^{\prime}\in \bigcup_{i=1}^n L_i$ and $q^{\prime}\in \bigcup_{i=1}^nH_i$, or $\{p^{\prime},q^{\prime}\}\nsubseteq \bigcup_{i=1}^n L_i \cup \bigcup_{i=1}^n H_i$ and $(\lambda(p^{\prime}),\lambda(q^{\prime}))\in D$.
\end{itemize}
$$\infer[(\mbox{Res1})]{P\backslash I \xrightarrow[\lambda]{p:\alpha\cdot({\vec L})}P^{\prime}\backslash I}{P \xrightarrow[\lambda]{p:\alpha\cdot({\vec L})}P^{\prime}~~{\sf symb}(\alpha)\notin I\cup\overline{I}} \quad
\infer[(\mbox{Res2})]{P\backslash I \xrightarrow[\lambda]{\tau}P^{\prime}\backslash I}{P \xrightarrow[\lambda]{\tau}P^{\prime}}$$
$$\infer[\mbox{(Con)}]{A({\vec v})\xrightarrow[\lambda]{\delta}P^{\prime}}{P\{{\vec v}/{\vec x}\}\xrightarrow[\lambda]{\delta}P^{\prime}~~A({\vec x})\stackrel{\rm def}{=}P}
$$
  \caption{Single-labelled Transitions}\label{LTS1}
\end{figure}

Obviously, $\{\tau,\tau\}$ is pairwise unrelated. We use pairwise unrelated multisets to characterize multi-labelled transitions.
Roughly speaking,
given two multi-labelled transitions $P\xrightarrow[\lambda_1]{\Delta_1}P^{\prime}$ and $Q\xrightarrow[\lambda_2]{\Delta_2}Q^{\prime}$ with $\Delta_1$ and $\Delta_2$ being pairwise unrelated respectively, there exists a transition of $P\oplus_D Q$ if $\Delta_1 \uplus \Delta_2$ are pairwise unrelated. Moreover, if there exist some communications between $P$ and $Q$, then the transition should take them into account (see $\Delta_m^{\tau}$ and $\Delta_0$ in Definition \ref{Com2}). The multi-labelled transition rule (Com2) for parallel composed processes, which is a general version of rule (Com1), is defined as follows.
\begin{definition}[\bf Multi-labelled Transitions]\label{Com2}
$$\infer[(\mbox{\rm Com2})]{P\oplus_D Q\xrightarrow[\lambda]{\Delta} P^{\prime}\oplus_{D^{\prime}} Q^{\prime}}{
    P\xrightarrow[\lambda_1]{\Delta_1}P^{\prime}\quad Q\xrightarrow[\lambda_2]{\Delta_2}Q^{\prime}\quad {\sf PUnrel}(\Delta_1)\quad {\sf PUnrel}(\Delta_2)
  \quad {\sf PUnrel}(\Delta_1\uplus \Delta_2)
  }$$
where: $\Delta = \Delta^{\prime} \uplus \Delta_m^{\tau}$ with $\Delta^{\prime} = (\Delta_1\uplus \Delta_2)\backslash\!\!\backslash\Delta_0$, $m = \frac{{\sf size}(\Delta_0)}{2}$ and $\Delta_0 = \{p:\alpha\cdot(\vec{L}),q:\overline{\alpha}\cdot
        (\vec{M})\mid (p,q)\in D \mbox{ and } p:\alpha\cdot(\vec{L})\in \Delta_1
        \mbox{ and } q:\overline{\alpha}\cdot(\vec{M})\in \Delta_2\}$; $\lambda:|P^{\prime}|\cup |Q^{\prime}|\rightarrow |P|\cup|Q|$ is defined by $\lambda(p)=\lambda_1(p)$ for $p\in |P^{\prime}|$ and $\lambda(q)=\lambda_2(q)$ for $q\in |Q^{\prime}|$; and $(p^{\prime},q^{\prime})\in D^{\prime} \subseteq |P^{\prime}|\times |Q^{\prime}|$, if
        \begin{itemize}
        \item either $p^{\prime}\in \bigcup_{i=1}^n L_{i}$ and $q^{\prime}\in \bigcup_{i=1}^n M_{i}$, $p:\alpha\cdot(\vec{L})\in \Delta_0$, $q:\overline{\alpha}\cdot(\vec{M})\in \Delta_0$ and $(p,q)\in D$; ({\it communication between $\alpha$ and $\overline{\alpha}$})
        \item  or $(\lambda(p^{\prime}),\lambda(q^{\prime}))\in D$ and $\{p^{\prime},q^{\prime}\}\nsubseteq \bigcup_{i=1}^n L_{i} \cup \bigcup_{i=1}^n M_{i}$ for any $p:\alpha\cdot(\vec{L})\in \Delta_0$, $q:\overline{\alpha}\cdot(\vec{M})\in \Delta_0$ and $(p,q)\in D$. ({\it inheritance})
        \end{itemize}
\end{definition}

In $P\xrightarrow[\lambda]{\Delta}P^{\prime}$, when ${\sf size}(\Delta)= 1$, we use the unique element to represent the multiset.
We can easily extend multi-labelled transitions to canonical processes.

\begin{example}[Multi-labelled transitions]
Let $f_1,g_1\in \Sigma_1$ and $f_2,g_2\in \Sigma_2$. Consider processes
$P = f_1(x)\cdot (\overline{g_1}(x)\cdot(\ast)) \oplus f_2(y)\cdot(\ast,\ast)$ and $Q = \overline{f_1}(1)\cdot (\ast) \oplus \overline{f_2}(2)\cdot(\ast, \ast)$, where $|P|=\{1,2\}$ such that ${\sf cs}(P(1)) = f_1(x)\cdot (\overline{g_1}(x)\cdot(\ast))$ and ${\sf cs}(P(2)) = f_2(y)\cdot(\ast,\ast)$, and $|Q|=\{3,4\}$ such that ${\sf cs}(Q(3))=\overline{f_1}(1)\cdot (\ast)$ and ${\sf cs}(Q(4))=\overline{f_2}(2)\cdot(\ast,\ast)$.
We have
\begin{center}
$P\xrightarrow[\lambda_1]
 {\{1:f_11\cdot(L_1),2:f_22\cdot(L_2,L_3)\}}P^{\prime}$\quad
and\quad
$Q\xrightarrow[\lambda_2]
{\{3:\overline{f_1}1\cdot(L_4),4:\overline{f_2}2\cdot(L_5,L_6)\}}Q^{\prime}$
\end{center}
where $P^{\prime} = \overline{g_1}(1)\cdot(\ast)\oplus (\ast\oplus \ast)$,
$Q^{\prime} = \ast\oplus (\ast\oplus\ast)$, $L_1=|\overline{g_1}(1)\cdot(\ast)|$,
$L_2 = |\ast|$, $L_3 =|\ast|$, $L_4 = |\ast|$, $L_5 = |\ast|$ and $L_6 = |\ast|$.
The graph of $P^{\prime}$ has $3$ vertices. Let
$|P^{\prime}|= \{1,6,7\}$ with $L_1=\{1\}$, $L_2=\{6\}$ and $L_3=\{7\}$, $P^{\prime}(1)=\overline{g_1}(1)\cdot(\ast)$, $P^{\prime}(6)= \ast$ and $P^{\prime}(7)=\ast$, then we have $\lambda_1(1)=1$ and $\lambda_1(6)=\lambda_1(7)=2$. Similarly, let $|Q^{\prime}|=\{3,8,9\}$ with $L_4=\{3\}$, $L_5=\{8\}$ and $L_6=\{9\}$, $Q^{\prime}(3)= \ast$, $Q^{\prime}(8) = \ast$ and $Q^{\prime}(9)=\ast$, then we have $\lambda_2(3) = 3$ and $\lambda_2(8)=\lambda_2(9) = 4$.
If $D=\{(1,3),(2,4)\}$, then we get
$P\oplus_D Q\xrightarrow[\lambda]{\{\tau,\tau\}}
P^{\prime}\oplus_{D^{\prime}} Q^{\prime}$
where $D^{\prime} = \{(1,3), (6,8), (6,9), (7,8), (7,9)\}$ and $\lambda= \lambda_1\circ \lambda_2$.
\qed\end{example}

In LLTSs, unrelated actions could occur consecutively and the order of their occurrences does not affect the final process. Moreover, a multi-labelled transition can be realized by a sequence of single-labelled transitions with the labels appearing in the multiset.
\begin{lemma}[\bf Diamond Property]\label{diomand}
\begin{enumerate}
  \item If $P\xrightarrow[\lambda_1]{\delta_1}P^{\prime}$, $Q\xrightarrow[\lambda_2]{\delta_2}Q^{\prime}$ and $P\oplus_D Q\xrightarrow[\lambda]{\{\delta_1,\delta_2\}}
   P^{\prime}\oplus_{D^{\prime}}Q^{\prime}$ (i.e. $\{\delta_1,\delta_2\}$ is pairwise unrelated), then we have
   $P\oplus_D Q\xrightarrow[\mu_1]{\delta_1}P^{\prime}\oplus_{D_1}Q
   \xrightarrow[\mu_2]{\delta_2}
   P^{\prime}\oplus_{D^{\prime}}Q^{\prime}$, $P\oplus_D Q\xrightarrow[\rho_1]{\delta_2}P\oplus_{D_2}Q^{\prime}
   \xrightarrow[\rho_2]{\delta_1}
   P^{\prime}\oplus_{D^{\prime}}Q^{\prime}$ and
   $\mu_1\circ \mu_2 = \rho_1\circ \rho_2 = \lambda$.
  \item Given a process $P$, if $P\xrightarrow[\lambda]{\Delta}P^{\prime}$ then there exist $P_0 = P$, $P_n = P^{\prime}$, and $P_i \xrightarrow[\lambda_{i+1}]{\delta_{i+1}}P_{i+1}$ such that $P_0\xrightarrow[\lambda_1]{\delta_1}P_1
      \xrightarrow[\lambda_2]{\delta_2}\cdots
    \xrightarrow[\lambda_{n}]{\delta_{n}}P_n$,
   where $n={\sf size}(\Delta)$, $\delta_{i+1}\in \Delta$ with $i\in \{0,\ldots,n-1\}$ and $\lambda= \lambda_1\circ\cdots\circ\lambda_n$.
\end{enumerate}
\end{lemma}

We write $P \xrightarrow[\lambda]{\tau^{\ast}}P^{\prime}$ if there exists
$n\geq 1$ such that $P = P_1$, $P^{\prime} = P_n$, $P_1\xrightarrow[\lambda_1]{\tau}P_2\xrightarrow[\lambda_2]{\tau}\cdots$ $\xrightarrow[\lambda_{n-1}]{\tau}P_n$ and $\lambda = \lambda_1\circ\lambda_2\circ \cdots \circ \lambda_{n-1}$.
Let $\widehat{\Delta}$ represent the multiset with all the invisible labels (i.e. $\tau$s) removed from $\Delta$.
By diamond property, if $P\xrightarrow[\lambda]{\Delta}P^{\prime}$, then we can get $P\xrightarrow[\lambda_1]{\tau^{\ast}}P_1
\xrightarrow[\lambda_2]{\widehat{\Delta}}P_2\xrightarrow[\lambda_3]{\tau^{\ast}}P^{\prime}$ and $\lambda=\lambda_1\circ\lambda_2\circ\lambda_3$ for some processes $P_1$ and $P_2$. $P\xLongrightarrow[\lambda,\lambda_1,\lambda^{\prime}]
                 {\widehat{\Delta}}P^{\prime}$ means that there exist
                  processes $P_1$ and $P_1^{\prime}$ such that
                 $P\xrightarrow[\lambda]{\tau^{\ast}}P_1\xrightarrow[\lambda_1]
                 {\widehat{\Delta}}P_1^{\prime}\xrightarrow[\lambda^{\prime}]
                {\tau^{\ast}}P^{\prime}$.

\subsection{Weak Bisimulation}
In this part, we define an early version of weak bisimulation on ${\sf Proc}$ through triples $(P,E,Q)$ by taking locations into account, where $E \subseteq |P|\times|Q|$ specifies the pairs of locations as well as the pairs of corresponding subprocesses to be considered together.

\begin{definition}[\bf Localized Relation \cite{Ehrhard2013ccts}]
A localized relation on ${\sf Proc}$ is a set $\mathcal{R}\subseteq {\sf Proc}\times\mathcal{P}({\sf Loc}^2)\times{\sf Proc}$ such that, if $(P,E,Q)\in \mathcal{R}$ then $E\subseteq |P|\times|Q|$. $\mathcal{R}$ is symmetric if $(P,E,Q)\in \mathcal{R}$ then $(Q,^t\!\!E,P)\in \mathcal{R}$, where $^t\!E=\{(q,p)\mid (p,q)\in E\}$.
\end{definition}

\begin{definition}[\bf Corresponding Multiset]
Given $\widehat{\Delta}$ (only containing visible labels), a corresponding multiset for $\widehat{\Delta}$, denoted by $\widehat{\Delta}^{c}$, is a multiset of labels such that for any $p:\alpha\cdot(\vec{L})\in \widehat{\Delta}$ there exists a unique $q:\alpha\cdot(\vec{M})\in \widehat{\Delta}^{c}$ (with the same $\alpha$), and vice versa.
\end{definition}

\begin{definition}[\bf Weak Bisimulation]
A symmetric localized relation $\mathcal{S}$ is a (localized early) weak bisimulation such that:
\begin{itemize}
  \item if $(P,E,Q)\in\mathcal{S}$ and $P \xrightarrow[\lambda]{\tau}P^{\prime}$ then
   $Q\xrightarrow[\rho]{\tau^{\ast}}Q^{\prime}$ with $(P^{\prime},E^{\prime},Q^{\prime})\in\mathcal{S}$ for some $E^{\prime}\subseteq|P^{\prime}|\times|Q^{\prime}|$ such that, if $(p^{\prime},q^{\prime})\in E^{\prime}$ then $(\lambda(p^{\prime}),\rho(q^{\prime}))\in E$;
  \item if $(P,E,Q)\in \mathcal{S}$ and
  $P\xlongrightarrow[\lambda]{\widehat{\Delta}}P^{\prime}$
   then
  $Q\xLongrightarrow[\rho,\rho_1,\rho^{\prime}]
  {\widehat{\Delta}^{c}}Q^{\prime}$ with the conditions that for any $p:\alpha\cdot({\vec L})\in \widehat{\Delta}$ there exists $q:\alpha\cdot({\vec M})\in \widehat{\Delta}^{c}$ such that $(p,\rho(q))\in E$ and $(P^{\prime},E^{\prime},Q^{\prime})\in \mathcal{S}$ for some $E^{\prime}\subseteq|P^{\prime}|\times|Q^{\prime}|$ such that if $(p^{\prime},q^{\prime})\in E^{\prime}$ then $(\lambda(p^{\prime}),\rho\rho_1\rho^{\prime}(q^{\prime}))\in E$.
\end{itemize}
\end{definition}

\begin{definition}
$P$ and $Q$ are weakly bisimilar, written $P\approx Q$, if there exist a weak bisimulation $\mathcal{R}$ and a relation $E\subseteq|P|\times|Q|$ such that $(P,E,Q)\in\mathcal{R}$.
\end{definition}
\noindent{\it Remark}. The localized relation is needed later for the parallel composition in the proof that weak bisimilarity is a congruence. Though there are infinitely many possible corresponding multisets of a given $\widehat{\Delta}$, we only need to fix one in the definition of weak bisimulation.

\section{Characterizations}\label{Soundness}
\subsection{Soundness}
For soundness, we need to prove that weak bisimilarity implies weak barbed congruence.
We first have to prove that weak bisimilarity is a congruence, i.e. $\approx$ is an equivalence relation and is preserved by the structures of VCCTS.
\begin{proposition}\label{equivalence}
$\approx$ is an equivalence relation.
\end{proposition}

The relationship between two bisimulations is the following proposition.
\begin{proposition}\label{propsitionbisimulationbarb}
If $P\approx Q$ then $P\stackrel{\bullet}{\approx}Q$.
\end{proposition}

\begin{theorem}\label{bisimilationCongruence}
 $\approx$ is a congruence.
\end{theorem}

Here, the main challenge is to extend the localized relation $\mathcal{R}$ to another localized relation $\mathcal{R}^{\prime}$ to handle the parallel composition in VCCTS. 
In CCS \cite{Milner1989}, if $\mathcal{R}$ is a weak bisimulation and $P~\mathcal{R}~Q$, then we can prove that $S\mid P$ and $S\mid Q$ are weakly bisimilar just by proving that a new relation $\mathcal{R}^{\prime}$ extending $\mathcal{R}$, such that $(S\mid P)~\mathcal{R}^{\prime}~(S\mid Q)$, is a weak bisimulation. However, we cannot simply do this in VCCTS. Moreover we have to record the locations of the subprocesses and the edges of locations which represent the possible communications between subprocesses.
To overcome this obstacle in VCCTS, we use $S\oplus_C P$ to specify the parallel composition of $S$ and $P$ with some $C\subseteq|S|\times|P|$. Similarly, we say that $S\oplus_D Q$ with some relation $D\subseteq |S|\times|Q|$ is a parallel composition of $S$ and $Q$. 
Then we prove the parallel extensions, $S\oplus_C P$ and $S\oplus_D Q$, are weakly bisimilar. The proof depends on locations and the relations $C$ and $D$.

An easy consequence of Theorem \ref{bisimilationCongruence} and Proposition \ref{propsitionbisimulationbarb} is the following.
\begin{theorem}[\bf Soundness]\label{maintheorem}
If $P\approx Q$, then $P\cong Q$.
\end{theorem}

\subsection{Completeness}
Inspired by \cite{Sangiorgi:2011}, the proof below requires the image-finite condition to use the stratification of weak bisimilarity. We adopt the method in \cite{Sangiorgi:2011}, and pay more attentions to the $n$-ary symbols, locations and non-sequential semantics.
\begin{definition}
An LLTS is image-finite if for all $P$, $\lambda$ and $\Delta$, the set $\{P^{\prime}\mid P\xLongrightarrow[\lambda]{\widehat{\Delta}}P^{\prime}\}$ is finite;
an LLTS is finitely-branching if it is image-finite, and for each $P$, the set $\{\Delta\mid P\xLongrightarrow[\lambda]{\widehat{\Delta}}P^{\prime} \hbox{ for some } P^{\prime} \}$ is finite.
\end{definition}

\begin{definition}[\bf Stratification of Weak Bisimilarity]
\begin{itemize}
  \item $\approx_0 \stackrel{\rm def}{=} {\sf Proc}\times\mathcal{P}({\sf Loc}^2)\times{\sf Proc}$;
  \item $(P,E,Q)\in \approx_{n+1}$, for $n\ge 0$,
        \begin{itemize}
          \item if $P \xrightarrow[\lambda]{\tau}P^{\prime}$ then
   $Q\xrightarrow[\rho]{\tau^{\ast}}Q^{\prime}$ with $(P^{\prime},E^{\prime},Q^{\prime})\in \approx_{n}$ for some $E^{\prime}\subseteq|P^{\prime}|\times|Q^{\prime}|$ such that, if $(p^{\prime},q^{\prime})\in E^{\prime}$ then $(\lambda(p^{\prime}),\rho(q^{\prime}))\in E$;
          \item if $P\xlongrightarrow[\lambda]{\widehat{\Delta}}P^{\prime}$ then
  $Q\xLongrightarrow[\rho,\rho_1,\rho^{\prime}]
  {\widehat{\Delta}^{c}}Q^{\prime}$ with the conditions that for any $p:\alpha\cdot({\vec L})\in \widehat{\Delta}$ there exists $q:\alpha\cdot({\vec M})\in \widehat{\Delta}^{c}$ such that $(p,\rho(q))\in E$ and $(P^{\prime},E^{\prime},Q^{\prime})\in \approx_{n}$ for some $E^{\prime}\subseteq|P^{\prime}|\times|Q^{\prime}|$ such that if $(p^{\prime},q^{\prime})\in E^{\prime}$ then $(\lambda(p^{\prime}),\rho\rho_1\rho^{\prime}(q^{\prime}))\in E$;
        \item the converses of the above two cases also hold;
        \end{itemize}

  \item $\approx_{\omega} \stackrel{\rm def}{=}\bigcap_{n\ge 0} \approx_n$
\end{itemize}
\end{definition}

\begin{lemma}\label{two_bisim}
On finitely-branching LLTSs, $\approx$ and $\approx_{\omega}$ coincide.
\end{lemma}

By definition, if $P\approx Q$ then $(P,|P|\times|Q|,Q)\in \approx$.
It holds in this paper by adopting the relaxed communication constraints.
So in image-finite systems, if $(P,|P|\times|Q|,Q)\notin \approx_{\omega}$, then $P~/\!\!\!\!\!\!\approx Q$.
The key of the proof is to show that for any image-finite canonical processes $P$ and $Q$, if they are not weakly bisimilar, then $P\mid R$ and $Q\mid R$ are not weakly barbed bisimilar for some canonical process $R$.

\begin{theorem}\label{completeness}
Suppose that for $n\ge 0$, $(P,|P|\times|Q|,Q)\notin\approx_n $ and $P$, $Q$ are image-finite. Then there is a canonical process $R$ such that one of the following holds:
\begin{itemize}
  \item $P\mid R ~ /\!\!\!\!\!\!\stackrel{\bullet}{\approx} Q^{\prime} \mid R$ for all $Q^{\prime}$ such that $Q\xrightarrow[]{}^{\ast} Q^{\prime}$;
  \item $P^{\prime}\mid R ~ /\!\!\!\!\!\!\stackrel{\bullet}{\approx} Q \mid R$ for all $P^{\prime}$ such that $P\xrightarrow[]{}^{\ast} P^{\prime}$.
\end{itemize}
\end{theorem}

A straightforward consequence of Theorem \ref{maintheorem} and Theorem \ref{completeness} is the following.
\begin{theorem}[\bf Characterization]
For any image-finite canonical processes $P$ and $Q$, $P\cong Q$ iff $P\approx Q$.
\end{theorem}


\section{Discussions on Top-down Tree Automata and VCCS}\label{Discussion}
\subsection{Top-down Tree Automata}\label{top_down_automata}
\begin{definition}
The set of $\Sigma$-trees with value passing is the smallest set such that
\begin{itemize}
  \item $\ast()\cdot()$ (simply written $\ast$) is a $\Sigma$-tree with value passing,
  \item if $t_1,\ldots,t_n$ are $\Sigma$-trees with value passing, $f\in \Sigma_n$ and $x\in {\bf Var}$, then $f(x)\cdot(t_1,\ldots,t_n)$ is a $\Sigma$-tree with value passing.
\end{itemize}
\end{definition}

We recall the definition of top-down tree automata \cite{tata2007} without the part of initial states, i.e. $\mathcal{A}=(\mathcal{Q},\Sigma,\mathcal{T})$.
Let $\mathcal{Q}$ be a finite subset of $\mathcal{K}$. $(Q,f(x),(Q_1,\ldots,Q_n))$ represents the transition $Q(f(x)\cdot(t_1,\ldots,t_n))\rightarrow f(x)\cdot(Q_1(t_1),\ldots,Q_n(t_n))$,
where $f\in \Sigma_n (n\ge 1)$, $x\in {\bf Var}$, $Q,Q_1,\ldots,Q_n \in \mathcal{Q}$ and $t_1,\ldots,t_n$ are $\Sigma$-trees with value passing. Therefore, $f(x)\cdot(t_1,\ldots,t_n)$ is {\it recognized} by $\mathcal{A}$ at state $Q$, if there exists $(Q,f(x),(Q_1,\ldots,Q_n))\in \mathcal{T}$ and $t_i$ is recognized by $\mathcal{A}$ at state $Q_i$ for $i\in \{1,\ldots,n\}$.


Given $\mathcal{A}= (\mathcal{Q},\Sigma,\mathcal{T})$, for any state $Q\in\mathcal{Q}$, we define a process $\langle\mathcal{A}\rangle_Q$ for the pair $(\mathcal{A},Q)$. We define $\langle\mathcal{A}\rangle_Q^{\mathcal{X}}$ in which $\mathcal{X}$ is a finite subset of $\mathcal{K}$ and $\mathcal{X}$ is the set of processes that have been defined. Let $\langle\mathcal{A}\rangle_Q = \langle\mathcal{A}\rangle_Q^{\emptyset}$.
\begin{itemize}
  \item If $Q\notin\mathcal{X}$, then $\langle\mathcal{A}\rangle_Q^{\mathcal{X}}\stackrel{\rm def}{=} S$ where $S$ is the sum of prefixed processes $f(x)\cdot(\langle\mathcal{A}\rangle_{Q_1}^{\mathcal{X}\cup\{Q\}},\ldots,$ $\langle\mathcal{A}\rangle_{Q_n}^{\mathcal{X}\cup\{Q\}})$ for each $(Q,f(x),(Q_1,\ldots,Q_n))\in \mathcal{T}$.
  \item If $Q\in\mathcal{X}$, then $\langle\mathcal{A}\rangle_Q^{\mathcal{X}} \stackrel{\rm def}{=} Q$.
\end{itemize}
We can check that ${\sf cs}(\langle\mathcal{A}\rangle_Q)$ is the sum of processes $f(x)\cdot(\langle\mathcal{A}\rangle_{Q_1}^{\mathcal{X}\cup\{Q\}},\ldots,\langle\mathcal{A}\rangle_{Q_n}^{\mathcal{X}\cup\{Q\}})$ for all $(Q,f(x),(Q_1,\ldots,Q_n))\in \mathcal{T}$. For each $\Sigma$-tree with value passing $t = f(x)\cdot(t_1,\ldots,t_n)$, we define ${\sf proc}(t) = \overline{f}(v)\cdot({\sf proc}(t_1),\ldots,{\sf proc}(t_n))$, for some $v\in {\bf Val}$, and ${\sf proc}(\ast)=\ast$. Moreover, a process $G\langle\Phi\rangle$ is an idle process if $\Phi(p)= \ast$ for any $p\in|G|$.
\begin{proposition}\label{tree_automata}
Let $\mathcal{A}=(\mathcal{Q},\Sigma,\mathcal{T})$ be a top-down tree automaton, let $Q\in \mathcal{Q}$ and let $t$ be a $\Sigma$-tree with value passing. If $t$ is recognized by $\mathcal{A}$ at state $Q$, then $\langle\mathcal{A}\rangle_Q\mid{\sf proc}(t)$ can reduce to an idle process.
\end{proposition}

However, the other direction is not satisfied. We show an example for this.
\begin{example}[Counterexample]\label{Counterexample}
$\mathcal{A}=(\mathcal{Q},\Sigma,\mathcal{T})$ is a top-down tree automaton, $\mathcal{Q}= \{Q,Q_1,Q_2,Q_{11},Q_{12},Q_{21},Q_{22}\}$, $\Sigma=\{f,g_1,g_2\}$, and $\mathcal{T} = \{(Q,f(x),(Q_1,Q_2)),$ $(Q_1,g_1(x),(Q_{11},Q_{12})),(Q_2,g_2(x),(Q_{21},Q_{22}))\}$, see Fig. \ref{fig:automata_tree}. Let $t = f(x)\cdot(g_1(x)\cdot(g_2(x)\cdot(\ast,\ast),\ast),\ast)$, see Fig. \ref{fig:automata_tree}.
We have $\langle\mathcal{A}\rangle_Q = f(x)\cdot(g_1(x)\cdot(\ast,\ast),g_2(x)\cdot(\ast,\ast))$ and ${\sf proc}(t) = \overline{f}(1)\cdot(\overline{g_1}(1)\cdot(\overline{g_2}(1)\cdot(\ast,\ast),\ast),\ast)$.
Process $f(x)\cdot(g_1(x)\cdot(\ast,\ast),g_2(x)\cdot(\ast,\ast))\mid \overline{f}(1)\cdot(\overline{g_1}(1)\cdot(\overline{g_2}(1)\cdot(\ast,\ast),\ast),\ast)$ can reduce to an idle process. But $t$ cannot be recognized by $\mathcal{A}$ at $Q$.
\qed\end{example}

\begin{figure}[!tb]
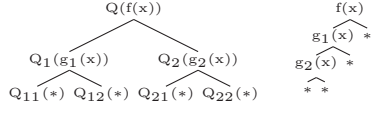
\tiny
\centering
$\Tree [.Q(f(x)) [.Q_1(g_1(x)) Q_{11}($\ast$) Q_{12}($\ast$) ] [.Q_2(g_2(x)) Q_{21}($\ast$)  Q_{22}($\ast$) ] ]\qquad
\Tree [.f(x) [.g_1(x) [.g_2(x) $\ast$ $\ast$  ] $\ast$  ]  $\ast$  ] $
\caption{A Tree Automaton (the Left One) and a Tree (the Right One)}\label{fig:automata_tree}
\end{figure}
\subsection{Embedding Value-passing CCS in VCCTS}
In Milner's CCS book \cite{Milner1989}, a version of CCS with value passing (VCCS) has been studied and VCCS could be translated into pure CCS, i.e. an interleaving semantics was given to VCCS.
Therefore, similar notations of VCCS can be derived from CCS immediately (cf. \cite{Milner1989} for more details).
In this part, we intend to embed VCCS in VCCTS and a non-sequential semantics is given to VCCS with locations, like \cite{boudol1994theory,castellani2001process}.

Here we assume that $\Sigma_n = \emptyset$ for $n>1$ and all the associated graphs are complete.
We can build all the VCCS processes over $\overline{\Sigma}$ (defined in Section \ref{Syntax}) using the composition rules of VCCTS.
Every process $P$ is of the shape $G\langle\Phi\rangle\backslash I$ ($I$ is possibly empty) and $G$ is a complete graph which means that
all the subprocesses parallel composed can communicate with each other, similar to the location versions of CCS \cite{boudol1994theory,castellani2001process}. Moreover, the recursive canonical guarded sums in VCCTS coincide with the guarded sums in VCCS, and every process $P$ is of the form $(S_1\mid \cdots\mid S_n)\backslash I$ with each $S_i$ being a recursive canonical guarded sum. At last, a fully abstract non-sequential semantics
for VCCS can be directly derived from the theory of VCCTS in this paper.

\begin{example}\label{example_law}
Compared to Milner's expansion law \cite{Milner1989}, e.g. $(a\mid b) = a.b + b.a$, one can easily check that $\overline{f}(1)\cdot({\bf 0})\mid \overline{g}(2)\cdot({\bf 0})~~ /\!\!\!\!\!\!\approx~\overline{f}(1)\cdot(\overline{g}(2)\cdot({\bf 0}))+ \overline{g}(2)\cdot(\overline{f}(1)\cdot({\bf 0}))$ due to multi-labelled transitions.
\qed\end{example}

\section{The Alternating Bit Protocol}\label{Exmaples}
Alternating Bit Protocol (ABP) is a simple data link layer network protocol.
ABP is used when a transmitter $P_1$ wants to send messages to a receiver $P_2$, with the assumptions that the channel may corrupt a message and that $P_1$ and $P_2$ can decide whether or not they have received a correct message.
Each message from $P_1$ to $P_2$ contains a data part and a one-bit sequence number, i.e. a value that is $0$ or $1$. $P_2$ can send two acknowledge messages, i.e. $(Ack,0)$ and $(Ack,1)$, to $P_1$.

\begin{figure}[!tp]\small
  \centering
  \begin{tabular}{rcl}
  $P_1(t_1,b)$ & $\stackrel{\rm def}{=}$ & ${\bf if}~{\sf null}(t_1)~{\bf then}~\overline{\sf send}((End,b))\cdot({\sf ack}(x)\cdot($  \\
   &  &\qquad ${\bf if}~ x = (Ack,b)~{\bf then}~{\bf 0}~{\bf else}~\overline{f}(0)\cdot(P_1(t_1,b))))$ \\
   &  &${\bf else}~\overline{\sf send}(({\sf head}(t_1),b))\cdot({\sf ack}(x)\cdot( $\\
   &  &\qquad ${\bf if}~x=(Ack,b)~{\bf then}~\overline{f}(0)\cdot(P_1({\sf tail}(t_1),\neg b))$\\
   &  &\qquad\qquad\qquad\qquad~~${\bf else}~\overline{f}(0)\cdot(P_1(t_1,b))))$\\
  $P_2(t_2,b) $ &$\stackrel{\rm def}{=}$ & ${\sf send}(x)\cdot({\bf if}~{\sf snd}(x) = b ~$ \\
   &    & \quad\qquad ${\bf then} ~({\bf if}~{\sf fst}(x) = End~{\bf then}~Succ(t_2)$\\
   &   & \qquad\qquad\qquad\qquad${\bf else}~\overline{\sf ack}((Ack,b))\cdot(P_2({\sf append}(t_2,{\sf fst}(x)),\neg b))) $ \\
   &  & \quad\qquad ${\bf else}~\overline{\sf ack}(Ack,\neg b)\cdot(P_2(t_2,b)) )$\\
  $A$ &$\stackrel{\rm def}{=}$ & $f(x)\cdot(A)$\\
\end{tabular}
\caption{The Alternating Bit Protocol in VCCTS}\label{ABP}
\end{figure}

In Fig. \ref{ABP} we provide a specification of ABP in VCCTS.
${\sf send, ack}\in \Sigma_1$ are used to transform messages.
The transmitter $P_1$ has a list $t_1$ containing the messages to be sent, and the receiver $P_2$ also has a list $t_2$ containing the received messages. The list is equipped with operations ${\sf head}$ (returning the head of a list), ${\sf tail}$ (returning a list with the first element removed), ${\sf append}$ (inserting an element as the last element of the new list) and ${\sf null}$ (testing whether a list is empty).
We use ${\sf fst}$ (and ${\sf snd}$) to return the first element (and the second element) of a pair.
$End$ is a sentry to indicate that all the messages in $t_1$ have been transformed. $Succ(t_2)$ indicates that the receiver has successfully received all the messages. We use an auxiliary process $A$ to interact with $P_1$, because we have to specify ABP in a canonical form.

In \cite{Milner1989}, ABP was formalized in CCS with a interleaving semantics.
In \cite{lanese2010operational}, ABP was investigated in a broadcasting semantics, and the authors wrote $n[P]^c_{l,r}$ for a node named $n$, located at location $l$, executing $P$, synchronized on channel $c$ and with the transmission radius $r$. Two nodes can communicate if they are in the radius of each other with the same synchronized channels. In this paper, we intend to {\it emphasize} that graphs can concisely characterize communicating capacities (i.e. topology of connections) in a non-sequential semantics. For instance, we use $\oplus$ to specify that $Q$ cannot communicate with processes $P_1$, $P_2$ and $A$ in Proposition \ref{ABP_prop}, while interactions are usually blocked by restrictions, e.g. in CCS.

\begin{proposition}\label{ABP_prop}
For any processes $Q$,
$((A\mid P_1(t,b))\mid P_2([],b))\oplus Q\xrightarrow[]{}^{\ast}((A\mid{\bf 0})\mid Succ(t))\oplus Q^{\prime}$ for some process $Q^{\prime}$ such that $Q\xrightarrow[]{}^{\ast}Q^{\prime}$, which means that if the transmitter $P_1$ and the receiver $P_2$ can communicate with each other but they cannot communicate with $Q$, then the system can reach a state where the messages in the transmitter are correctly received by the receiver no matter what happens in $Q$.
\end{proposition}


\section{Related Work}\label{Related_work}
For lack of space, we just discuss some closely related work.
The interested reader is referred to \cite{boudol2008twenty} for a detailed survey of concurrent theories with non-sequential semantics.

The theories of CCTS \cite{Ehrhard2013ccts,Ehrhard2015dpc} and VCCTS \cite{liu2016vccts} are certainly the most related to the present work, and
the differences have been discussed in Section \ref{Intro}. 

Graph-based process calculi can also be found in \cite{ferrari2006synchronised,konig2001observational}, where processes were regarded as graphs and their evolutions were described by graph rewriting rules. The topology and connection structure of these systems are represented in terms of nodes and edges, similar to our calculus. But there are some differences between them and our calculus.
In \cite{konig2001observational}, a notion of bisimulation for graph rewriting systems is introduced.
However, the bisimulation is too coarse-grained and there is a lack of true concurrency semantics to handle the degree of parallelism of a system.
In \cite{ferrari2006synchronised}, distributed computing was modeled based on Synchronised Hyperedge Replacement, but lacked a theory of behavioural equivalences. 

There are different approaches to endow CCS with a non-sequential semantics:
\begin{itemize}
\item In \cite{boudol1988non}, Boudol and Castellani proposed a proved transition system for CCS, in which single-labelled transitions were identified by proofs. Then, a partial-order multiset LTS was extracted by equivalence permutations of single-labelled transitions, preserving causal relations and concurrent relations.
\item In \cite{boudol1994theory,castellani2001process},
    localities were introduced to describe explicitly the distribution of processes,
    either from a statical approach where locations are assigned to process statically before processes are executed,
    or from a dynamical approach where they are assigned dynamically when executions proceed.
    Then, a transition carried both a single action and a string of locations standing for the accessing path.
    However, during the process evolution, the string of locations might be either totally discarded,
    or partially recorded.
\item In \cite{degano1990partial},
    Degano et al. proposed an operational semantics for CCS via a partial-order derivation relation.
    The derivation relation was defined on sets of sequential subprocesses of CCS, called grapes,
    and described the actions of the sequential subprocesses and the causal dependencies among them.
\end{itemize}
Compared to them, we propose an LLTS, labelled by multisets of labels and residual functions, for VCCTS which has a richer topological structure. Canonical guarded sums in VCCTS play the same role as grapes in \cite{degano1990partial}.
Compared to \cite{boudol1994theory,castellani2001process}, we use locations to identify dynamically the distribution of processes and use residual functions to keep track of the full information of locations during process evolution. 

Behavioural equivalence is an important idea of process calculi, and it equates processes that have the same behaviours.
Milner's CCS book \cite{Milner1989} is a milestone for bisimilarity employed to define behavioural equivalences and to reason about them.
Subsequently, different bisimulation-based equivalences have been proposed to CCS-like languages based on various LTSs, such as location bisimulation \cite{boudol1994theory,castellani2001process}, partial ordering observational equivalence \cite{degano1990partial}, distributed bisimulation \cite{castellani1989distributed}, etc.
Meanwhile, barbed congruence \cite{Milner1992barbed} is another important behavioural equivalence based on reduction semantics different from LTSs. However, as we know, none of \cite{boudol1994theory,castellani2001process,degano1990partial,castellani1989distributed} provided labelled characterizations of reduction barbed congruence like \cite{Milner1992barbed,sangiorgi1993expressing,Sangiorgi:2011}
We give a labelled characterization of reduction barbed congruence in a weak version similar to \cite{Sangiorgi:2011}, but with a richer topology and a non-sequential semantics.

\section{Conclusions}\label{Conclusions}
In this paper, we propose a fully abstract semantics for VCCTS. We have developed two kinds of operational semantics, different from the ones in \cite{liu2016vccts}, for VCCTS, i.e. a reduction semantics and a labelled transition semantics.
We have studied different forms of behaviour equivalences, i.e. weak barbed congruence and weak bisimilarity, and have proved that the two relations coincide. In this paper, the relaxed communication constraints are essential to the completeness.
We have discussed top-down tree automata and value-passing CCS in VCCTS.

The ABP example suggests that VCCTS could express protocols in wireless networks based on broadcast communication, and this is our ongoing work.
We also would like to study the addition of probabilities to the calculus and to apply it to concurrent scenarios involving probability.
\bibliographystyle{splncs03}
\bibliography{myref}

\begin{thebibliography}{10}
\providecommand{\url}[1]{\texttt{#1}}
\providecommand{\urlprefix}{URL }

\bibitem{boudol1988non}
Boudol, G., Castellani, I.: A non-interleaving semantics for {CCS} based on
  proved transitions. Fundamenta Informaticae  (1988)

\bibitem{boudol1994theory}
Boudol, G., Castellani, I., Hennessy, M., Kiehn, A.: A theory of processes with
  localities. Formal Aspects of Computing  6(2),  165--200 (1994)

\bibitem{boudol2008twenty}
Boudol, G., Castellani, I., Hennessy, M., Nielsen, M., Winskel, G.: Twenty
  years on: Reflections on the {CEDISYS} project. {C}ombining true concurrency
  with process algebra. In: Concurrency, Graphs and Models, pp. 757--777.
  Springer (2008)

\bibitem{castellani2001process}
Castellani, I.: Process algebras with localities. Handbook of Process Algebra
  (2001)

\bibitem{castellani1989distributed}
Castellani, I., Hennessy, M.: Distributed bisimulations. Journal of the ACM
  (JACM)  36(4),  887--911 (1989)

\bibitem{tata2007}
Comon, H., Dauchet, M., Gilleron, R., L\"oding, C., Jacquemard, F., Lugiez, D.,
  Tison, S., Tommasi, M.: Tree automata techniques and applications. Available
  on: \url{http://www.grappa.univ-lille3.fr/tata} (2007), release October, 12th
  2007

\bibitem{darondeau1990causal}
Darondeau, P., Degano, P.: Causal trees: Interleaving+ causality. In:
  Proceedings of the LITP Spring School. pp. 239--255. Springer-Verlag (1990)

\bibitem{degano1990partial}
Degano, P., De~Nicola, R., Montanari, U.: A partial ordering semantics for
  {CCS}. Theoretical Computer Science  75(3),  223--262 (1990)

\bibitem{Ehrhard2013ccts}
Ehrhard, T., Jiang, Y.: {CCS} for trees (2013),
  \url{http://arxiv.org/abs/1306.1714}

\bibitem{Ehrhard2015dpc}
Ehrhard, T., Jiang, Y.: A dendroidal process calculus (2015),
  \url{https://www.irif.univ-paris-diderot.fr/~ehrhard/}

\bibitem{ferrari2006synchronised}
Ferrari, G.L., Hirsch, D., Lanese, I., Montanari, U., Tuosto, E.: Synchronised
  hyperedge replacement as a model for service oriented computing. In: Formal
  Methods for Components and Objects. pp. 22--43. Springer (2006)

\bibitem{Hoare1985communicating}
Hoare, C.: Communicating sequential processes. Prentice-hall (1985)

\bibitem{kiehn1994comparing}
Kiehn, A.: Comparing locality and causality based equivalences. Acta
  Informatica  31(8),  697--718 (1994)

\bibitem{konig2001observational}
K{\"o}nig, B., Montanari, U.: Observational equivalence for synchronized graph
  rewriting with mobility. In: Theoretical Aspects of Computer Software. pp.
  145--164. Springer (2001)

\bibitem{lanese2010operational}
Lanese, I., Sangiorgi, D.: An operational semantics for a calculus for wireless
  systems. Theoretical Computer Science  411(19),  1928--1948 (2010)

\bibitem{liu2016vccts}
Liu, S., Jiang, Y.: Value-passing {CCS} for trees: a theory for concurrent
  systems. In: Theoretical Aspects of Software Engineering. Accepted (2016)

\bibitem{Milner1989}
Milner, R.: Communication and Concurrency. Prentice-Hall, Inc., Upper Saddle
  River, NJ, USA (1989)

\bibitem{milner1992picalculus}
Milner, R., Parrow, J., Walker, D.: A calculus of mobile processes, {I}.
  Information and computation  100(1),  1--40 (1992)

\bibitem{Milner1992barbed}
Milner, R., Sangiorgi, D.: Barbed bisimulation. In: Automata, Languages and
  Programming, pp. 685--695. Springer (1992)

\bibitem{reisig1985petri}
Reisig, W.: Petri nets: an introduction, volume 4 of {EATCS} monographs on
  theoretical computer science (1985)

\bibitem{sangiorgi1993expressing}
Sangiorgi, D.: Expressing mobility in process algebras: first-order and
  higher-order paradigms  (1993)

\bibitem{Sangiorgi:2011}
Sangiorgi, D.: Introduction to Bisimulation and Coinduction. Cambridge
  University Press, New York, NY, USA (2011)

\bibitem{winskel1989introduction}
Winskel, G.: An introduction to event structures. Linear Time, Branching Time
  and Partial Order in Logics and Models for Concurrency pp. 364--397 (1989)

\end{thebibliography}

\newpage
\appendix
\renewcommand{\appendixname}{Appendix~\Alph{section}}

\section{Appendix}
The proof for Lemma \ref{lemma:cs}.
\begin{proof}
It is easy by induction on $R$.
\qed\end{proof}
The proof for Lemma \ref{lemma:barbbisim}.
\begin{proof}
We need to prove that $\stackrel{\bullet}{\approx}$ is reflexive, symmetric and transitive. It is straightforward from the definition.
\qed\end{proof}
The proof for Proposition \ref{prop:congruence}.
\begin{proof}For the first statement,
from the definition of congruence, it is obvious that the identity relation contained in $\mathcal{R}$ is a congruence. And congruences are closed under arbitrary unions and contexts.

For the second statement, let $\mathcal{E}$ be a congruence defined by  $(P,Q)\in \mathcal{E}$ if and only if for any $Y$-context $R$ one has $(R[P/Y],R[Q/Y])\in \mathcal{R}$. Therefore, $\mathcal{E}$ is a congruence contained in $\mathcal{R}$ (because we can take $R = Y$) and hence $\mathcal{E}\subseteq \overline{\mathcal{R}}$. Conversely, let $(P,Q)\in \overline{\mathcal{R}}$ and $R$ be a $Y$-context. Because $\overline{\mathcal{R}}$ is a congruence, we have $(R[P/Y],R[Q/Y])\in \overline{\mathcal{R}}$. We have $(R[P/Y],R[Q/Y])\in \mathcal{R}$ from $\overline{\mathcal{R}}\subseteq \mathcal{R}$ by definition of $\overline{\mathcal{R}}$ and hence $(P,Q)\in \mathcal{E}$.
\qed\end{proof}

\subsection{Proofs for Diamond Property}
\begin{proof} For the case (1),
from $P\oplus_D Q\xrightarrow[\lambda]{\{\delta_1,\delta_2\}}
P^{\prime}\oplus_{D^{\prime}}Q^{\prime}$, we have that $\{\delta_1,\delta_2\}$ is unrelated.

From $P\xrightarrow[\lambda_1]{\delta_1}P^{\prime}$ and $Q\xrightarrow[\lambda_2]{\delta_2}Q^{\prime}$, we know $\lambda_1:|P^{\prime}|\rightarrow |P|$ and $\lambda_2:|Q^{\prime}|\rightarrow|Q|$, respectively. From $P\oplus_D Q\xrightarrow[\lambda]{\{\delta_1,\delta_2\}}
P^{\prime}\oplus_{D^{\prime}}Q^{\prime}$, we know $\lambda:|P^{\prime}|\cup |Q^{\prime}|\rightarrow|P|\cup |Q|$ which is consistent with $\lambda_1$ and $\lambda_2$, and $(p^{\prime},q^{\prime})\in D^{\prime}$ if $(\lambda(p^{\prime}),\lambda(q^{\prime}))\in D$.

For $P\oplus_D Q \xrightarrow[\mu_1]{\delta_1}P^{\prime}\oplus_{D_1}Q
\xrightarrow[\mu_2]{\delta_2}
P^{\prime}\oplus_{D_1^{\prime}}Q^{\prime}$, we have $\mu_1:|P^{\prime}|\cup |Q|\rightarrow |P|\cup |Q|$ and $\forall p^{\prime}\in |P^{\prime}|, q\in |Q|, \mu_1(p^{\prime}) = \lambda_1(p^{\prime})$ and $\mu_1(q)=q$. $(p^{\prime},q)\in D_1$ if $(\mu_1(p^{\prime}),\mu_1(q))\in D$, that is $(\lambda_1(p^{\prime}),q)\in D$. We also have $\mu_2:|P^{\prime}|\cup |Q^{\prime}|\rightarrow |P^{\prime}|\cup |Q|$ and $\forall p^{\prime}\in |P^{\prime}|, q^{\prime}\in |Q^{\prime}|, \mu_2(p^{\prime})=p^{\prime}$ and $\mu_2(q^{\prime})=\lambda_2(q^{\prime})$.
$(p^{\prime},q^{\prime})\in D_1^{\prime}$
if $(\mu_2(p^{\prime}),\mu_2(q^{\prime}))\in D_1$, that is $(p^{\prime},\lambda_2(q^{\prime}))\in D_1$. So $(p^{\prime},q^{\prime})\in D_1^{\prime}$, if $(\mu_1\circ \mu_2(p^{\prime}), \mu_1\circ \mu_2(q^{\prime}))\in D$ that is $(\mu_1(p^{\prime}),\mu_2(q^{\prime}))\in D$, i.e. $ (\lambda_1(p^{\prime}),\lambda_2(q^{\prime}))
=(\lambda(p^{\prime}),\lambda(q^{\prime}))\in D$ from the definitions of $\mu_1$ and $\mu_2$.

Therefore, $D_1^{\prime}=D^{\prime}$, and we have
$P\oplus_D Q\xrightarrow[\mu_1]{\delta_1}P^{\prime}\oplus_{D_1}Q
\xrightarrow[\mu_2]{\delta_2}
P^{\prime}\oplus_{D^{\prime}}Q^{\prime}$ with $\mu_1\circ\mu_2 = \lambda$.

For the case $P\oplus_D Q\xrightarrow[\rho_1]{\delta_2}
P\oplus_{D_2}Q^{\prime}\xrightarrow[\rho_2]{\delta_1}
P^{\prime}\oplus_{D_2^{\prime}}Q^{\prime}$, it is similar.

For the case (2), induction on the size of $\Delta$ with the assumption that $\Delta$ is pairwise unrelated.
\qed\end{proof}
\subsection{Localized Early Weak Bisimulation Is an Equivalence}\label{app1}
\begin{lemma}\label{taus}
Let $\mathcal{R}$ be a localized early weak bisimulation. If $(P,E,Q)\in \mathcal{R}$ and $P\xrightarrow[\lambda]{\tau^{\ast}}P^{\prime}$, then $Q\xrightarrow[\rho]{\tau^{\ast}}Q^{\prime}$ with $(P^{\prime},E^{\prime},Q^{\prime})\in \mathcal{R}$ for some $E^{\prime}\subseteq |P^{\prime}|\times|Q^{\prime}|$ such that if $(p^{\prime},q^{\prime})\in E^{\prime}$ then $(\lambda(p^{\prime}),\rho(q^{\prime}))\in E$.
\end{lemma}
\begin{proof}
Induction on the length of the derivation of $P\xrightarrow[\lambda]{\tau^{\ast}}P^{\prime}$.
\qed\end{proof}
\begin{lemma}\label{generals}
If $P\xrightarrow[\lambda]{\tau^{\ast}}P_1$,
$P_1\xLongrightarrow[\lambda_1,\lambda_2,\lambda_3]
{\widehat{\Delta}}P_1^{\prime}$ and $P_1^{\prime}\xrightarrow[\lambda^{\prime}]{\tau^{\ast}}
P^{\prime}$, then
$P\xLongrightarrow[\lambda\lambda_1,\lambda_2,\lambda_3\lambda^{\prime}]
{\widehat{\Delta}}P^{\prime}$.
\end{lemma}
\begin{proof}
Straightforward.
\qed\end{proof}
\begin{lemma}\label{lammebisimulation}
A symmetric localized relation $\mathcal{R}\subseteq {\sf Proc}\times\mathcal{P}({\sf Loc}^2)\times{\sf Proc}$ is a localized early weak bisimulation if and only if the following properties hold:
\begin{itemize}
  \item If $(P,E,Q)\in \mathcal{R}$ and $P\xLongrightarrow[\lambda,\lambda_1,\lambda^{\prime}]
      {\widehat{\Delta}}P^{\prime}$, then $Q\xLongrightarrow[\rho,\rho_1,\rho^{\prime}]
      {\widehat{\Delta}^c}Q^{\prime}$ with for any pair of labels $p:\alpha\cdot
      ({\vec L})\in \widehat{\Delta}$ and $q:\alpha\cdot({\vec M})\in \widehat{\Delta}^c$, we have $(\lambda(p),\rho(q))\in E$ and $(P^{\prime},E^{\prime},Q^{\prime})\in \mathcal{R}$ for some $E^{\prime}\subseteq |P^{\prime}|\times|Q^{\prime}|$ such that if $(p^{\prime},q^{\prime})\in E^{\prime}$ then $(\lambda\lambda_1\lambda^{\prime}(p^{\prime}),
       \rho\rho_1\rho^{\prime}(q^{\prime}))\in E$.
  \item If $(P,E,Q)\in \mathcal{R}$ and $P\xrightarrow[\lambda]{\tau^{\ast}}P^{\prime}$, then
      $Q\xrightarrow[\rho]{\tau^{\ast}}Q^{\prime}$ with $(P^{\prime},E^{\prime},Q^{\prime})\in \mathcal{R}$ for some $E^{\prime}\in |P^{\prime}|\times|Q^{\prime}|$ such that if $(p^{\prime},q^{\prime})\in E^{\prime}$ then $(\lambda(p^{\prime}),\rho(q^{\prime}))\in E$.
\end{itemize}
\end{lemma}
\begin{proof}
($\Leftarrow$) Because $\xrightarrow[\lambda]{\tau}$ and $\xrightarrow[\lambda]{\widehat{\Delta}}$ are special cases of
$\xrightarrow[\lambda]{\tau^{\ast}}$ and $\xLongrightarrow[\lambda,\lambda_1,\lambda^{\prime}]
{\widehat{\Delta}}$ respectively, this direction is obvious.

($\Rightarrow$) For the first statement, assume that $(P,E,Q)\in \mathcal{R}$ and $P\xLongrightarrow[\lambda,\lambda_1,\lambda^{\prime}]
{\widehat{\Delta}}P^{\prime}$ which is $P\xrightarrow[\lambda]{\tau^{\ast}}P_1
\xrightarrow[\lambda_1]{\widehat{\Delta}}
P_1^{\prime}\xrightarrow[\lambda^{\prime}]{\tau^{\ast}}P^{\prime}$, by Lemma \ref{taus} we can get $Q\xrightarrow[\rho]{\tau^{\ast}}Q_1$ with $(P_1,E_1,Q_1)\in \mathcal{R}$ where $E_1$ satisfies the property that if $(p_1,q_1)\in E_1$ then $(\lambda(p_1),\rho(q_1))\in E$.

From $P_1\xrightarrow[\lambda_1]{\widehat{\Delta}}P_1^{\prime}$ and $(P_1,E_1,Q_1)\in \mathcal{R}$, we can get $Q_1\xLongrightarrow[\rho_1,\rho_2,\rho_1^{\prime}]
{\widehat{\Delta}^c}Q_1^{\prime}$ with the condition that for any pair of labels $p:\alpha\cdot({\vec L})\in \widehat{\Delta}$ and $q:\alpha\cdot({\vec M})\in \widehat{\Delta}^c$ we have $(p,\rho_1(q)) \in E_1$ and $(P_1^{\prime},E_1^{\prime},Q_1^{\prime})\in \mathcal{R}$, where $E_1^{\prime}$ satisfies that if $(p_1^{\prime},q_1^{\prime})\in E_1^{\prime}$ then $(\lambda_1(p_1^{\prime}),\rho_1\rho_2\rho_1^{\prime}(q_1^{\prime}))\in E_1$.

Since $P_1^{\prime}\xrightarrow[\lambda^{\prime}]{\tau^{\ast}}
P^{\prime}$ and $(P_1^{\prime}, E_1^{\prime},Q_1^{\prime})\in \mathcal{R}$,
by Lemma \ref{taus}, we can have $Q_1^{\prime}\xrightarrow[\rho^{\prime}]{\tau^{\ast}}Q^{\prime}$ with $(P^{\prime},E^{\prime},Q^{\prime})\in \mathcal{R}$ where $E^{\prime}$ satisfies that if $(p^{\prime},q^{\prime})\in E^{\prime}$ then $(\lambda^{\prime}(p^{\prime}),\rho^{\prime}(q^{\prime}))\in E_1^{\prime}$.

With $Q\xrightarrow[\rho]{\tau^{\ast}}Q_1$, $Q_1\xLongrightarrow[\rho_1,\rho_2,\rho_1^{\prime}]
{\widehat{\Delta}^c}Q_1^{\prime}$  and $Q_1^{\prime}\xrightarrow[\rho^{\prime}]{\tau^{\ast}}Q^{\prime}$, by Lemma \ref{generals} we can get
$Q\xLongrightarrow[\rho\rho_1,\rho_2,\rho_1^{\prime}\rho^{\prime}]
{\widehat{\Delta}^c}Q^{\prime}$. Meanwhile, we have $(P^{\prime},E^{\prime},Q^{\prime})\in \mathcal{R}$. The conditions on residual functions are satisfied obviously.

For the second statement, it is straightforward from Lemma \ref{taus}.
\qed\end{proof}
\begin{lemma}[\bf Reflexivity]\label{reflexivity}
Let $\mathcal{I}$ be the localized relation defined by $(P,E,Q)\in \mathcal{I}$ if $P=Q$ and $E={\rm Id}_{|P|}$. Then $\mathcal{I}$ is a localized early weak bisimulation.
\end{lemma}
\begin{proof}
Straightforward.
\qed\end{proof}

Let $\mathcal{R}$ and $\mathcal{S}$ be localized relations. We define a localized relation $\mathcal{S}\circ\mathcal{R}$ for the composition of $\mathcal{R}$ and $\mathcal{S}$. $(P,H,R)\in \mathcal{S}\circ\mathcal{R}$ if $H\subseteq |P|\times|R|$ and there exist $Q$, $E$ and $F$ such that $(P,E,Q)\in \mathcal{R}$, $(Q,F,R)\in \mathcal{S}$ and $F\circ E\subseteq H$.
\begin{lemma}[\bf Transitivity]\label{transitivity}
If $\mathcal{R}$ and $\mathcal{S}$ are localized early weak bisimulations, then $\mathcal{S}\circ\mathcal{R}$ is also a localized early weak bisimulation.
\end{lemma}
\begin{proof}
Obviously, $\mathcal{S}\circ\mathcal{R}$ is symmetric.
Then the proof just follows the definition of the localized early weak bisimulation using the Lemma \ref{lammebisimulation}.

From the hypothesis, let $(P,E,Q)\in \mathcal{R}$, $(Q,F,R)\in \mathcal{S}$ and $(P,H,R)\in\mathcal{S}\circ\mathcal{R}$ with $F\circ E \subseteq H$.

(1) If $P\xLongrightarrow[\lambda,\lambda_1,\lambda^{\prime}]
{\widehat{\Delta}}P^{\prime}$, then $Q\xLongrightarrow[\rho,\rho_1,\rho^{\prime}]
{\widehat{\Delta}^c}Q^{\prime}$ and for any pair of labels $p:\alpha\cdot({\vec L})\in \widehat{\Delta}$ and $q:\alpha\cdot({\vec M})\in \widehat{\Delta}^c$ we have
$(\lambda(p),\rho(q))\in E$ and $(P^{\prime},E^{\prime},Q^{\prime})$ $\in \mathcal{R}$ with $E^{\prime}$ such that if $(p^{\prime}, q^{\prime})\in E^{\prime}$ then $(\lambda\lambda_1\lambda^{\prime}(p^{\prime}),\rho\rho_1
\rho^{\prime}(q^{\prime}))\in E$.
From $(Q,F,R)\in \mathcal{S}$ and $Q\xLongrightarrow[\rho,\rho_1,\rho^{\prime}]
{\widehat{\Delta}^c}Q^{\prime}$, we have $R\xLongrightarrow[\sigma,\sigma_1,\sigma^{\prime}]
{(\widehat{\Delta}^c)^c}R^{\prime}$ and for any pair of labels $q:\alpha\cdot({\vec M})\in \widehat{\Delta}^c$ and $r:\alpha\cdot({\vec N})\in(\widehat{\Delta}^c)^c$ with $(\rho(q),\sigma(r))\in F$ and $(Q^{\prime},F^{\prime},R^{\prime})\in \mathcal{S}$ with $F^{\prime}$ such that if $(q^{\prime},r^{\prime})\in F^{\prime}$ then $(\rho\rho_1\rho^{\prime}(q^{\prime}),\sigma\sigma_1\sigma^{\prime}(r^{\prime}))\in F$.

Therefore, for any pair of labels $p:\alpha\cdot({\vec L})\in \widehat{\Delta}$ and $q:\alpha\cdot({\vec M})\in \widehat{\Delta}^c$ and pair of labels $q:\alpha\cdot({\vec M})\in \widehat{\Delta}^c$ and $r:\alpha\cdot({\vec N})\in(\widehat{\Delta}^c)^c$, we have $(\lambda(p),\sigma(r))$ $\in F\circ E\subseteq H$.
Let $H^{\prime}=\{(p^{\prime},r^{\prime})\in |P^{\prime}|\times|R^{\prime}|\mid (\lambda\lambda_1\lambda^{\prime}(p^{\prime}),
 \sigma\sigma_1\sigma^{\prime}(r^{\prime}))\in H  \}$.
It is obvious that $(P^{\prime},H^{\prime},R^{\prime})$ satisfies the residual conditions from the definition of $H^{\prime}$.
Next, we have to prove $F^{\prime}\circ E^{\prime}\subseteq H^{\prime}$, then show that $(P^{\prime}, H^{\prime}, R^{\prime})\in \mathcal{S}\circ \mathcal{R}$.
If $(p^{\prime},r^{\prime})\in F^{\prime}\circ E^{\prime}$, then there exist $(p^{\prime},q^{\prime})\in E^{\prime}$ and $(q^{\prime},r^{\prime})\in F^{\prime}$. We have to prove $(p^{\prime},r^{\prime})\in H^{\prime}$.

Since $(\lambda\lambda_1\lambda^{\prime}(p^{\prime}),\rho\rho_1\rho^{\prime}
(q^{\prime}))\in E$ and $(\rho\rho_1\rho^{\prime}(q^{\prime}),
\sigma\sigma_1\sigma^{\prime}(r^{\prime}))\in F$, we can get
$(\lambda\lambda_1\lambda^{\prime}(p^{\prime}),
\sigma\sigma_1\sigma^{\prime}(r^{\prime}))\in F\circ E \subseteq H$.

(2) From $(P,E,Q)\in \mathcal{R}$, if $P\xrightarrow[\lambda]{\tau^{\ast}}P^{\prime}$, then we have $Q\xrightarrow[\rho]{\tau^{\ast}}Q^{\prime}$ and $(P^{\prime},E^{\prime},Q^{\prime})\in \mathcal{R}$ with some $E^{\prime}$ such that if $(p^{\prime},q^{\prime})\in E^{\prime}$ then $(\lambda(p^{\prime}),\rho(q^{\prime}))\in E$.
Since $(Q,F,R)\in \mathcal{S}$ and $Q\xrightarrow[\rho]{\tau^{\ast}}Q^{\prime}$, we have
$R\xrightarrow[\sigma]{\tau^{\ast}}R^{\prime}$ and $(Q^{\prime},F^{\prime},R^{\prime})\in \mathcal{S}$ with some $F^{\prime}$ such that if $(q^{\prime},r^{\prime})$ then $(\rho(q^{\prime}),\sigma(r^{\prime}))\in F$. We just get $(P^{\prime},F^{\prime}\circ E^{\prime}, R^{\prime})\in \mathcal{S}\circ \mathcal{R}$ as required. And it is obvious that $F^{\prime}\circ E^{\prime}$ satisfies the residual condition.
\qed\end{proof}
Proof for Proposition \ref{equivalence}.
\begin{proof}
Because $\approx$ is reflexive by Lemma \ref{reflexivity} , symmetric from the definition and transitive by Lemma \ref{transitivity}.
\qed\end{proof}
The proof for Proposition \ref{propsitionbisimulationbarb}.
\begin{proof}
Let $\mathcal{R}$ be a localized early weak bisimulation. Let $\mathcal{B}$ be a binary relation on processes defined by: $(P,Q)\in \mathcal{B}$ if there exists some $E\subseteq |P|\times|Q|$ such that $(P,E,Q)\in \mathcal{R}$. Then we have to prove that $\mathcal{B}$ is a weak bared bisimulation. First, we know that $\mathcal{B}$ is symmetric, because $\mathcal{R}$ is symmetric.\\
(1) Let $(P,Q)\in \mathcal{B}$. If $P\xrightarrow[]{}^{\ast}P^{\prime}$ which is $P\xrightarrow[\lambda]{\tau^{\ast}}P^{\prime}$ for some residual function $\lambda$. Because $\mathcal{R}$ is a localized early weak bisimulation, let $(P,E,Q)\in \mathcal{R}$ with some $E\subseteq |P|\times|Q|$.
We have $Q\xrightarrow[\rho]{\tau^{\ast}}Q^{\prime}$ (i.e. $Q\xrightarrow[]{}^{\ast}Q^{\prime}$) and $(P^{\prime},E^{\prime},Q^{\prime})\in \mathcal{R}$ from Lemma \ref{lammebisimulation}.
Meanwhile, for any $(p^{\prime},q^{\prime})\in E^{\prime}$, the residual function $\rho$ satisfies $(\lambda(p^{\prime}),\rho(q^{\prime}))\in E$ . So we have $(P^{\prime}, Q^{\prime})\in \mathcal{B}$.\\
(2) Let $(P,Q)\in \mathcal{B}$. If $P\xrightarrow[]{}^{\ast}P^{\prime}$ and $P^{\prime}\downarrow_{B}$,
then there exists a transition $P^{\prime}\xrightarrow[\lambda_1^{\prime}]{\widehat{\Delta}}P_1$, where $\widehat{\Delta}$ is a pairwise unrelated multiset of labels of the form $p^{\prime}:fv\cdot({\vec L})$ for each $f\in B$. Since $\mathcal{R}$ is a localized early weak bisimulation, let $(P,E,Q)\in \mathcal{R}$ for some $E\subseteq |P|\times|Q|$.
Then we have $Q\xrightarrow[\rho]{\tau^{\ast}}Q^{\prime}$ (i.e. $Q\xrightarrow[]{}^{\ast}Q^{\prime}$) for some residual function $\rho$, and $E^{\prime}\subseteq |P^{\prime}|\times|Q^{\prime}|$ such that $(P^{\prime},E^{\prime},Q^{\prime})\in \mathcal{R}$.
Because $\mathcal{R}$ is a localized early weak bisimulation and $P^{\prime}\xrightarrow[\lambda_1^{\prime}]{\widehat{\Delta}}P_1$, we have $Q^{\prime}\xLongrightarrow[\rho^{\prime},\rho_1,\rho_2]
{\widehat{\Delta}^c}Q_1$ which means $Q^{\prime}\xrightarrow[\rho^{\prime}]{\tau^{\ast}}Q_1^{\prime}$ (i.e. $Q^{\prime}\xrightarrow[]{}^{\ast}Q_1^{\prime}$) with $Q^{\prime}_1\downarrow_{B}$.
From $P^{\prime}\downarrow_B$ and the transition $P^{\prime}\xrightarrow[\lambda_1^{\prime}]{\widehat{\Delta}}P_1$ satisfying the constraints between $B$ and $\widehat{\Delta}$, we get that
$Q\rightarrow^{\ast}Q_1^{\prime}$ with $Q^{\prime}_1\downarrow_{B}$ as required.
\qed\end{proof}
\subsection{Localized Early Weak Bisimulation Is a Congruence}\label{app2}
%
We use $S\oplus_C P$ to specify the parallel composition of $S$ and $P$ with some $C\subseteq|S|\times|P|$. Similarly, we say that $S\oplus_D Q$ with some relation $D\subseteq |S|\times|Q|$ is a parallel composition of $S$ and $Q$. The relations $C$ and $D$ should satisfy some constraints.
\begin{definition}[\bf Adapted Triple of Relations \cite{Ehrhard2013ccts}]
We say that a triple of relations $(D,D^{\prime},E)$ with $D\subseteq A\times B$, $D^{\prime} \subseteq A\times B^{\prime}$ and $E\subseteq B\times B^{\prime}$ is {\it adapted}, if for any $(a,b,b^{\prime})\in A\times B\times B^{\prime}$ with $(b,b^{\prime})\in E$, $(a,b)\in D$  iff $(a,b^{\prime})\in D^{\prime}$.
\end{definition}

Let $\mathcal{R}$ be a localized relation on processes. We define a new localized relation on processes $\mathcal{R}^{\prime}$, by ensuring that $(U,F,V)\in \mathcal{R}^{\prime}$ and the following conditions are satisfied:
\begin{itemize}
  \item there exist a process $S$, a triple $(P,E,Q)\in \mathcal{R}$, $C\subseteq |S|\times|P|$ and $D\subseteq|S|\times|Q|$ such that $U=S\oplus_C P$ and $V= S\oplus_D Q$,
  \item $(C,D,E)$ is adapted,
  \item $F$ is the relation $(\mbox{Id}_{|S|}\cup E) \subseteq |U|\times|V|$.
\end{itemize}
We call that the relation $\mathcal{R}^{\prime}$ is a {\it parallel extension} of $\mathcal{R}$.
\begin{lemma}[\bf \cite{Ehrhard2013ccts}]\label{lammeextendRsymmtric}
If $R$ is symmetric, then its parallel extension $\mathcal{R}^{\prime}$ is also symmetric.
\end{lemma}
\begin{proposition}\label{propositionparallel}
If $\mathcal{R}$ is a localized early weak bisimulation, then its parallel extension $\mathcal{R}^{\prime}$ is also a localized early weak bisimulation.
\end{proposition}
\begin{proof}
We can get that $\mathcal{R}^{\prime}$ is symmetric from Lemma \ref{lammeextendRsymmtric}.

Let $(U,F,V)\in \mathcal{R}^{\prime}$ with $(P,E,Q)\in \mathcal{R}$, $U= S\oplus_C P$, $V= S\oplus_D Q$, $(C,D,E)$ is adapted and $F=\mbox{Id}_{|S|}\cup E$.\\
\textit{\textbf{Case of a $\tau$-transition.}} Given $U\xrightarrow[\lambda]{\tau}U^{\prime}$, we have to show $V\xrightarrow[\rho]{\tau^{\ast}}V^{\prime}$ with $(U^{\prime},F^{\prime},V^{\prime})\in \mathcal{R}^{\prime}$ such that for any $(u^{\prime},v^{\prime})\in F^{\prime}$ implies $(\lambda(u^{\prime}),\rho(v^{\prime}))\in F$. There are three cases for a $\tau${\it -transition} for $U=S\oplus_C P$. Meanwhile, we only focus on canonical processes. For the canonical guarded sum ${\sf cs}(P)$, its prefixed form ${\sf cs}(P)$ is of the form $pre\cdot(Q_1,\ldots,Q_n)+T$,
where $pre$ is a prefix, $T$ is a canonical guarded sum and $Q_1,\ldots,Q_n$ are canonical processes.\\
~\\
(1) The two locations are in $S$, and $S\oplus_C P\xrightarrow[\lambda]{\tau}S^{\prime}\oplus_{C^{\prime}}P$. If $s,t\in |S|$ with $s\frown_{S}t$ such that
${\sf cs}(S(s))= f(x)\cdot{\vec S} + \widetilde{S}$ and ${\sf cs}(S(t))=\overline{f}(e)\cdot{\vec T} + \widetilde{T}$ with ${\sf eval}(e)=v$, where $\widetilde{S}$ and $\widetilde{T}$ are canonical guarded sums. So we have $S\xrightarrow[\mu]{\tau}S^{\prime}$ with
\begin{itemize}
  \item $|S^{\prime}|=(|S|\setminus\{s,t\})\cup \bigcup_{i=1}^n|S_i\{v/x\}|\cup \bigcup_{i=1}^n|T_i|$
  \item and $\frown_{S^{\prime}}$ is the least symmetric relation on $|S^{\prime}|$ such that $s^{\prime}\frown_{S^{\prime}}t^{\prime}$ if $s^{\prime}\frown_{S_i\{v/x\}}t^{\prime}$, or $s^{\prime}\frown_{T_i}t^{\prime}$, or $(s^{\prime},t^{\prime})\in (\bigcup_{i=1}^n |S_i\{v/x\}|)\times(\bigcup_{i=1}^n|T_i|)$, or $\{s^{\prime},t^{\prime}\}\nsubseteq \bigcup_{i=1}^n|S_i\{v/x\}|\cup\bigcup_{i=1}^n|T_i|$ and $\mu(s^{\prime})\frown_S\mu(t^{\prime})$
\end{itemize}
where, $n$ is the arity of $f$, ${\vec S} = (S_1,\ldots,S_n)$ and ${\vec T} = (T_1,\ldots,T_n)$.
Note that $\mu$ is a residual function which is defined as:
$\mu(s^{\prime})= s$ if $s^{\prime}\in \bigcup_{i=1}^n|S_i\{v/x\}|$, $\mu(s^{\prime})= t$ if $s^{\prime}\in \bigcup_{i=1}^n|T_i|$, and $\mu(s^{\prime}) = s^{\prime}$ otherwise.
Then we have $U^{\prime}=S^{\prime}\oplus_{C^{\prime}} P$, where $C^{\prime}=\{(s^{\prime},p)\in |S^{\prime}|\times|P|\mid (\mu(s^{\prime}),p)\in C\}$ and $\lambda = \mu \cup \mbox{Id}_{|P|}$.

Similarly, for $V=S\oplus_D Q$ we have $V\xrightarrow[\rho]{\tau}V^{\prime}=S^{\prime}\oplus_{D^{\prime}}Q$ with $\rho = \mu \cup \mbox{Id}_{|Q|}$, and $D^{\prime}=\{(s^{\prime},q)\in |S^{\prime}|\times|Q|\mid (\mu(s^{\prime}),q)\in D\}$.

Then we have to show that the triple $(C^{\prime},D^{\prime},E)$ is adapted.
Let $s^{\prime}\in |S^{\prime}|$, $p\in |P|$ and $q\in |Q|$ such that $(p,q)\in E$. If $(s^{\prime},p)\in C^{\prime}$ then $(\mu(s^{\prime}),p)\in C$. Since $(C,D,E)$ is adapted, we have $(\mu(s^{\prime}),q)\in D$. So $(s^{\prime},q)\in D^{\prime}$. The converse is similar and we omit it here.

So we have $(U^{\prime},F^{\prime},V^{\prime})\in \mathcal{R}^{\prime}$ where $F^{\prime} = \mbox{Id}_{|S^{\prime}|}\cup E$. Then we check the residual condition. Given $(u^{\prime},v^{\prime})\in F^{\prime}$, either if $u^{\prime}=v^{\prime}\in |S^{\prime}|$ then $\lambda(u^{\prime})=\rho(v^{\prime})\in |S|$, or if $(u^{\prime},v^{\prime})\in E$ then $(\lambda(u^{\prime}),\rho(v^{\prime}))\in E$. So, in both cases we have $(\lambda(u^{\prime}),\rho(v^{\prime}))\in F$.

The symmetric case is similar, where we have $s,t\in |S|$ with $s\frown_{S}t$ such that ${\sf cs}(S(s))= \overline{f}(e)\cdot{\vec S} + \widetilde{S}$ with ${\sf eval}(e)=v$ and ${\sf cs}(S(t))=f(x)\cdot{\vec T} + \widetilde{T}$, where $\widetilde{S}$ and $\widetilde{T}$ are canonical guarded sums. \\
~\\
(2) The two locations are in $P$ and $S\oplus_C P\xrightarrow[\lambda]
{\tau} S\oplus_{C^{\prime}}P^{\prime}$. Let $p,r\in |P|$ with $p\frown_{P}r$ such that ${\sf cs}(P(p))=f(x)\cdot{\vec P}+ \widetilde{P}$ and
${\sf cs}(P(r))= \overline{f}(e)\cdot{\vec R} + \widetilde{R}$ with ${\sf eval}(e) = v$, where $\widetilde{P}$ and $\widetilde{R}$ are canonical guarded sums. So we have $P\xrightarrow[\mu]{\tau}P^{\prime}$ with
\begin{itemize}
  \item $|P^{\prime}|=(|P|\setminus\{p,r\})\cup \bigcup_{i=1}^n|P_i\{v/x\}|\cup \bigcup_{i=1}^n|R_i|$
  \item and $\frown_{P^{\prime}}$ is the least symmetric relation on $|P^{\prime}|$ such that $p^{\prime}\frown_{P^{\prime}}r^{\prime}$ if $p^{\prime}\frown_{P_i\{v/x\}}r^{\prime}$, or $p^{\prime}\frown_{R_i}r^{\prime}$, or $(p^{\prime},r^{\prime})\in (\bigcup_{i=1}^n |P_i\{v/x\}|)\times(\bigcup_{i=1}^n|R_i|)$, or $\{p^{\prime},r^{\prime}\}\nsubseteq \bigcup_{i=1}^n|P_i\{v/x\}|\cup\bigcup_{i=1}^n|R_i|$ and $\mu(p^{\prime})\frown_P\mu(r^{\prime})$
\end{itemize}
where, $n$ is the arity of $f$, ${\vec P} = (P_1,\ldots, P_2)$ and ${\vec R} = (R_1,\ldots,R_n)$.
$\mu$ is a residual function defined as:
$\mu(p^{\prime}) = p$ if $p^{\prime}\in \bigcup_{i=1}^n|P_i\{v/x\}|$, $\mu(p^{\prime}) = r$ if $p^{\prime}\in \bigcup_{i=1}^n|R_i|$, and $\mu(p^{\prime}) = p^{\prime}$ otherwise.
So we have $U^{\prime} = S\oplus_{C^{\prime}} P^{\prime}$ where $C^{\prime}= \{(s,p^{\prime})\in |S|\times|P^{\prime}|\mid (s,\mu(p^{\prime}))\in C\}$ and the residual function $\lambda = \mbox{Id}_{|S|}\cup \mu$.

Since $(P,E,Q)\in \mathcal{R}$, from $P\xrightarrow[\mu]{\tau}P^{\prime}$, we have $Q\xrightarrow[\nu]{\tau^{\ast}}Q^{\prime}$ with $(P^{\prime}, E^{\prime},Q^{\prime})\in \mathcal{R}$ where $E^{\prime}\subseteq|P^{\prime}|\times|Q^{\prime}|$ such that $(p^{\prime},q^{\prime})\in E^{\prime}$ implies $(\mu(p^{\prime}),\nu(q^{\prime}))\in E$. Let $D^{\prime}= \{(s,q^{\prime})\in |S|\times|Q^{\prime}|\mid (s,\nu(q^{\prime}))\in D\}$. From $V^{\prime} = S\oplus_{D^{\prime}}Q^{\prime}$, we have $V\xrightarrow[\rho]{\tau^{\ast}}V^{\prime}$ with $\rho= \mbox{Id}_{|S|}\cup \nu$.

Then we show that the triple $(C^{\prime},D^{\prime},E^{\prime})$ is adapted.
Let $(p^{\prime},q^{\prime})\in E^{\prime}$ and $s\in |S|$. If $(s,p^{\prime})\in C^{\prime}$, then we have $(s,\mu(p^{\prime}))\in C$. Since $(\mu(p^{\prime}),\nu(q^{\prime}))\in E$ and $(C,D,E)$ is adapted, we have $(s,\nu(q^{\prime}))\in D$. So we can get $ (s,q^{\prime})\in D^{\prime}$ from the definition of $D^{\prime}$. The other direction is similar.

So we have $(U^{\prime},F^{\prime}, V^{\prime})\in \mathcal{R}^{\prime}$ where $F^{\prime}=\mbox{Id}_{|S|}\cup E^{\prime}\subseteq |U^{\prime}|\times|V^{\prime}|$. Then we have to check the residual condition. Given $(u^{\prime},v^{\prime})\in F^{\prime}$, either $u^{\prime}=v^{\prime}\in |S|$ and then $\lambda(u^{\prime})=\rho(v^{\prime}) = u^{\prime}$, or $u^{\prime}\in |P^{\prime}|$, $v^{\prime}\in |Q^{\prime}|$ and $(u^{\prime},v^{\prime})\in E^{\prime}$ and then $(\lambda(u^{\prime}),\rho(v^{\prime}))=(\mu(u^{\prime}),\nu(v^{\prime}))\in E$. So we get $(\lambda(u^{\prime}),\rho(v^{\prime}))\in F$.

The symmetric case is similar, where $p,r\in |P|$ with $p\frown_{P}r$ such that ${\sf cs}(P(p))=\overline{f}(e)\cdot{\vec P}+ \widetilde{P}$  with ${\sf eval}(e) = v$ and ${\sf cs}(P(r))= f(x)\cdot{\vec R} + \widetilde{R}$, where $\widetilde{P}$ and $\widetilde{R}$ are canonical guarded sums.\\
~\\
(3) One of the locations from $S$ and the other from $P$, i.e. $S\oplus_C P \xrightarrow[\lambda]{\tau}S^{\prime}\oplus_{C^{\prime}}P^{\prime}$. Let $p\in |P|$ and $s\in |S|$ with $(s,p)\in C$. And we have ${\sf cs}(P(p))= f(x)\cdot{\vec P}+ \widetilde{P}$ and ${\sf cs}(S(s))= \overline{f}(e)\cdot{\vec S} + \widetilde{S}$ with ${\sf eval}(e) = v$, where $\widetilde{P}$ and $\widetilde{S}$ are canonical guarded sums.
Then we have $U^{\prime}= S^{\prime}\oplus_{C^{\prime}}P^{\prime}$ with $S^{\prime} = S[\oplus{\vec S}/s]$ and $P^{\prime} = P[\oplus{\vec P}\{v/x\}/p]$, where $n$ is the arity of $f$, ${\vec S} = (S_1,\ldots,S_n)$ and ${\vec P}\{v/x\} = (P_1\{v/x\},\ldots,P_n\{v/x\})$.

Let $C^{\prime}\subseteq|S^{\prime}|\times|P^{\prime}|$, and $(s^{\prime},p^{\prime})\in C^{\prime}$ if
$(\lambda(s^{\prime}),\lambda(p^{\prime}))\in C$
where, residual function $\lambda:|U^{\prime}|=|S^{\prime}|\cup|P^{\prime}|\rightarrow |U|= |S|\cup |P|$, and it is defined as follows: $\lambda(s^{\prime})= s$ if $s^{\prime}\in\bigcup_{i=1}^n |S_i|$, $\lambda(p^{\prime}) = p$ if $p^{\prime}\in \bigcup_{i=1}^n |P_i\{v/x\}|$ and $\lambda(u^{\prime})= u^{\prime}$ if $u^{\prime} \in (|S^{\prime}|\setminus \bigcup_{i=1}^n |S_i|) \cup (|P^{\prime}|\setminus \bigcup_{i=1}^n|P_i\{v/x\}|)$.

We have $P\xrightarrow[\lambda]{p:\alpha\cdot({\vec L})}P^{\prime}$ with $\alpha = fv$  such that $v$ is just the value received from $S$, and $L_i = |P_i\{v/x\}|$ for $i\in\{1,\ldots,n\}$. By $(P,E,Q)\in \mathcal{R}$, we have $Q\xLongrightarrow[\rho,\rho_1,\rho^{\prime}]{q:\alpha\cdot({\vec M})}Q^{\prime}$ with $(p,\rho(q))\in E$ and $(P^{\prime},E^{\prime},Q^{\prime})\in \mathcal{R}$ with $E^{\prime}$ such that $(p^{\prime},q^{\prime})\in E^{\prime}$ implies $(\lambda(p^{\prime}),\rho\rho_1\rho^{\prime}(q^{\prime}))\in E$.

We can decompose $Q\xLongrightarrow[\rho,\rho_1,\rho^{\prime}]{q:\alpha\cdot({\vec M})}Q^{\prime}$ as
$$ Q\xrightarrow[\rho]{\tau^{\ast}}Q_1\xrightarrow[\rho_1]
{q:\alpha\cdot({\vec M})}Q_1^{\prime}\xrightarrow[\rho^{\prime}]
{\tau^{\ast}}Q^{\prime}.$$
We have $V\xrightarrow[\mu]{\tau^{\ast}}V_1$ with $V_1=S\oplus_{D_1}Q_1$, $D_1=\{(s,q_1)\in |S|\times |Q_1|\mid (s,\rho(q_1))\in D\}$ and $\mu= \mbox{Id}_{|S|}\cup \rho$.

Since $(p,\rho(q))\in E$, $(s,p)\in C$ and $(C,D,E)$ is adapted, we have $(s,\rho(q))\in D$. So $(s,q)\in D_1$ from the definition of $D_1$.
We have $q\in Q_1$ with ${\sf cs}(Q_1(q))= f(x)\cdot{\vec R}+ \widetilde{R}$ and ${\sf cs}(S(s))= \overline{f}(e)\cdot{\vec S}+ \widetilde{S}$ with ${\sf eval}(e)=v$ where $v$ is the same value as the part of derivation for $S(s)$ in $U\xrightarrow[\lambda]{\tau}U^{\prime}$. Then we have $M_i = |R_i\{v/x\}|$ for $i\in \{1,\dots,n\}$.
We can get  $V_1\xrightarrow[\theta]{\tau}V_1^{\prime}=S^{\prime}\oplus_{D_1^{\prime}}
Q_1^{\prime}$ where $D_1^{\prime}\subseteq|S^{\prime}|\times|Q_1^{\prime}|$ which is defined as follows: given $(s^{\prime},q_1^{\prime})\in |S^{\prime}|\times|Q_1^{\prime}|$, we have $(s^{\prime},q_1^{\prime})\in D_1^{\prime}$ if $(\theta(s^{\prime}),\theta(q_1^{\prime}))\in D_1$,
and the residual function $\theta$ is defined by $\theta(v_1^{\prime})=v_1^{\prime}$ if $v_1^{\prime}\in (|S|\setminus\bigcup_{i=1}^n|S_i|)\cup
(|Q_1|\setminus\bigcup_{i=1}^n|R_i\{v/x\}|)$, $\theta(s^{\prime})=s$ if $s^{\prime}\in \bigcup_{i=1}^n|S_i|$ and $\theta(q_1^{\prime})= q_1$ if $q_1^{\prime}\in \bigcup_{i=1}^n|R_i\{v/x\}|$.

We also have $\theta(q_1^{\prime})=\rho_1(q_1^{\prime})$ for any $q_1^{\prime}\in |Q_1^{\prime}|$.

From $Q_1^{\prime}\xrightarrow[\rho^{\prime}]{\tau^{\ast}}Q^{\prime}$, we have $V_1^{\prime}=S^{\prime}\oplus_{D_1^{\prime}}Q_1^{\prime}
\xrightarrow[\mu^{\prime}]{\tau^{\ast}}V^{\prime}= S^{\prime}\oplus_{D^{\prime}}Q^{\prime}$ where $\mu^{\prime}=\mbox{Id}_{|S^{\prime}|}\cup \rho^{\prime}$ and
$D^{\prime}=\{(s^{\prime},q^{\prime})\in |S^{\prime}|\times |Q^{\prime}|\mid (s^{\prime},\rho^{\prime}(q^{\prime}))\in D_1^{\prime}\}$. So, we have $V\xrightarrow[\mu\theta\mu^{\prime}]{\tau^{\ast}}V^{\prime}$.
Let $F^{\prime}\subseteq |U^{\prime}|\times|V^{\prime}|$ be defined by $F^{\prime}= \mbox{Id}_{|S^{\prime}|}\cup E^{\prime}$. It is clear that
$(u^{\prime},v^{\prime})\in F^{\prime}$ implies $(\lambda(u^{\prime}),\mu\theta\mu^{\prime}(v^{\prime}))\in F$, since $(p^{\prime},q^{\prime})\in E^{\prime}$ implies $(\lambda(p^{\prime}),\rho\rho_1\rho^{\prime}(q^{\prime}))\in E$ and $\theta$ and $\rho_1$ coincide on $|Q_1^{\prime}|$.

Then we have to prove $(U^{\prime},F^{\prime},V^{\prime})\in \mathcal{R}^{\prime}$. To prove it, we can just show that the triple $(C^{\prime},D^{\prime},E^{\prime})$ is adapted. Let $s^{\prime}\in |S^{\prime}|$, $p^{\prime}\in |P^{\prime}|$ and $q^{\prime}\in |Q^{\prime}|$ with $(p^{\prime},q^{\prime})\in E^{\prime}$ (i.e. particularly $(\lambda(p^{\prime}),\rho\theta\rho^{\prime}(q^{\prime}))\in E$).

If $(s^{\prime},p^{\prime})\in C^{\prime}$, then we have to show that $(s^{\prime},q^{\prime})\in D^{\prime}$ which is $(s^{\prime},\rho^{\prime}(q^{\prime}))\in D_1^{\prime}$. Referring to the definition of $C^{\prime}$, we analyse it in three cases:
\begin{itemize}
  \item First case: $(s^{\prime},p^{\prime})\in (\bigcup_{i=1}^n|S_i|)\times(\bigcup_{i=1}^n|P_i\{v/x\}|)$.
       If $\rho^{\prime}(q^{\prime})\in \bigcup_{i=1}^n M_i = \bigcup_{i=1}^n |R_i\{v/x\}|$, then we have $(s^{\prime},\rho^{\prime}(q^{\prime}))\in D_1^{\prime}$ as required.
       If $\rho^{\prime}(q^{\prime})\notin \bigcup_{i=1}^n M_i$, then, for definition of $D_1^{\prime}$, we need to prove $(\theta(s^{\prime}),\rho\theta\rho^{\prime}(q^{\prime}))
      =(s,\rho\rho^{\prime}(q^{\prime}))\in D$. Since $(p^{\prime},q^{\prime})\in E^{\prime}$, we have $(\lambda(p^{\prime}),\rho\theta\rho^{\prime}(q^{\prime}))= (p,\rho\rho^{\prime}(q^{\prime}))\in E$. We also have $(s,p)\in C$, and hence $(s,\rho\rho^{\prime}(q^{\prime}))\in D$ as required for $(C,D,E)$ is adapted.
  \item Second case: $s^{\prime}\notin \bigcup_{i=1}^n|S_i|$. In order to prove $(s^{\prime},q^{\prime})\in D^{\prime}$, it suffices to prove that $(\theta(s^{\prime}),\rho\theta\rho^{\prime}(q^{\prime})) = (s^{\prime},\rho\theta\rho^{\prime}(q^{\prime}))\in D$. And we have $(s^{\prime},p^{\prime})\in C^{\prime}$ and $s^{\prime}\notin \bigcup_{i=1}^n |S_i|$, hence $(\lambda(s^{\prime}),\lambda(p^{\prime}))=
      (s^{\prime},\lambda(p^{\prime}))\in C$. Since $(p^{\prime},q^{\prime})\in E^{\prime}$, we have $(\lambda(p^{\prime}),\rho\theta\rho^{\prime}(q^{\prime}))\in E$. Thus we have $(s^{\prime},\rho\theta\rho^{\prime}(q^{\prime}))\in D$ since $(C,D,E)$ is adapted.
  \item Third case: $s^{\prime}\in \bigcup_{i=1}^n|S_i|$ and $p^{\prime}\notin \bigcup_{i=1}^n|P_i\{v/x\}|$, so we have $(\lambda(s^{\prime}),\lambda(p^{\prime}))=(s,p^{\prime})\in C$ (by definition of $C^{\prime}$ and $(s^{\prime},p^{\prime})\in C^{\prime}$). Since $(p^{\prime},q^{\prime})\in E^{\prime}$, if $\rho^{\prime}(q^{\prime})\notin \bigcup_{i=1}^n M_i$. To prove $(s^{\prime},\rho^{\prime}(q^{\prime}))\in D_1^{\prime}$, it suffices to check that $(\theta(s^{\prime}),\rho\theta\rho^{\prime}(q^{\prime})) = (s,\rho\rho^{\prime}(q^{\prime}))\in D$. It holds since $(C,D,E)$ is adapted, $(s,p^{\prime})\in C$ and $(p^{\prime},\rho\rho^{\prime}(q^{\prime}))\in E$ for $(p^{\prime},q^{\prime})\in E^{\prime}$.
      If $\rho^{\prime}(q^{\prime})\in \bigcup_{i=1}^n M_i$, then we have $(s^{\prime},\rho^{\prime}(q^{\prime}))\in (\bigcup_{i=1}^n|S_i|)\times (\bigcup_{i=1}^n M_i)$, so $(s^{\prime},\rho^{\prime}(q^{\prime}))\in D_1^{\prime}$.
\end{itemize}

Now we prove the converse. If $(s^{\prime},q^{\prime})\in D^{\prime}$, i.e. $(s^{\prime},\rho^{\prime}(q^{\prime}))\in D_1^{\prime}$, we have to show $(s^{\prime},p^{\prime})\in C^{\prime}$. We also consider three cases.
\begin{itemize}
  \item First case: $s^{\prime}\in (\bigcup_{i=1}^n|S_i|)$ and $\rho^{\prime}(q^{\prime})\in (\bigcup_{i=1}^n M_i)=(\bigcup_{i=1}^n |R_i\{v/x\}|)$. If $p^{\prime}\in (\bigcup_{i=1}^n L_i) = (\bigcup_{i=1}^n |P_i\{v/x\}|)$, then $(s^{\prime},p^{\prime})\in C^{\prime}$ as required.
      If $p^{\prime}\notin (\bigcup_{i=1}^n L_i)$, then $p^{\prime}\notin (\bigcup_{i=1}^n |P_i\{v/x\}|)$.
      Since $(p^{\prime},q^{\prime})\in E^{\prime}$, we have $(\lambda(p^{\prime}),\rho\theta\rho^{\prime}(q^{\prime}))\in E$, i.e. $(p^{\prime},\rho(q))\in E$. Since we have $(s^{\prime},q^{\prime})\in D^{\prime}$, we have $(\theta(s^{\prime}),\rho\theta\rho^{\prime}(q^{\prime}))\in D$, i.e.
      $(s,\rho(q))\in D$. So we have $(s,p^{\prime})\in C$ as $(C,D,E)$ is adapted. Since $(\lambda(s^{\prime}),\lambda(p^{\prime})) = (s,p^{\prime})\in C$ and $p^{\prime}\notin (\bigcup_{i=1}^n L_i)$, we have $(s^{\prime},p^{\prime})\in C^{\prime}$.

  \item Second case: $s^{\prime}\notin \bigcup_{i=1}^n|S_i|$. For the definition of $C^{\prime}$, it suffices to prove $(\lambda(s^{\prime}),\lambda(p^{\prime}))=
      (s^{\prime},\lambda(p^{\prime}))\in C$. Since $(s^{\prime},q^{\prime})\in D^{\prime}$ and $s^{\prime}\notin \bigcup_{i=1}^n|S_i|$, we have $(\theta(s^{\prime}),\rho\theta\rho^{\prime}(q^{\prime}))=
      (s^{\prime},\rho\theta\rho^{\prime}(q^{\prime}))\in D$. Since $(p^{\prime},q^{\prime})\in E^{\prime}$ we have $(\lambda(p^{\prime}),\rho\theta\rho^{\prime}(q^{\prime}))\in E$, and hence $(s^{\prime},\lambda(p^{\prime}))\in C$ for $(C,D,E)$ is adapted.
  \item Third case: $s^{\prime}\in \bigcup_{i=1}^n |S_i|$  and $\rho^{\prime}(q^{\prime})\notin\bigcup_{i=1}^n M_i$. If $p^{\prime}\notin\bigcup_{i=1}^n L_i$, then, to check $(s^{\prime},p^{\prime})\in C^{\prime}$, it suffices to prove
      $(\lambda(s^{\prime}),\lambda(p^{\prime}))=(s,p^{\prime})\in C$.
      We have $(s^{\prime},q^{\prime})\in D^{\prime}$ and hence $(\theta(s^{\prime}),\rho\theta\rho^{\prime}(q^{\prime})) =
      (s,\rho\rho^{\prime}(q^{\prime}))\in D$. Since $(p^{\prime},q^{\prime})\in E^{\prime}$, we have $(\lambda(p^{\prime}),\rho\theta\rho^{\prime}(q^{\prime}))=
      (p^{\prime},\rho\rho^{\prime}(q^{\prime}))\in E$ and hence $(s,p^{\prime})\in C$ for $(C,D,E)$ is adapted. If $p^{\prime}\in \bigcup_{i=1}^n L_i$, we have $(s^{\prime},p^{\prime})\in C^{\prime}$
      since $(s^{\prime},p^{\prime})\in (\bigcup_{i=1}^n|S_i|)\times(\bigcup_{i=1}^n|P_i\{v/x\}|)$.
\end{itemize}

The other case is similar, where $p\in |P|$ and $s\in |S|$ such that $(s,p)\in C$, ${\sf cs}(P(p))= \overline{f}(e)\cdot{\vec P}+ \widetilde{P}$ with ${\sf eval}(e) = v$ and ${\sf cs}(S(s))= f(x)\cdot{\vec S}+ \widetilde{S}$, and $\widetilde{P}$ and $\widetilde{S}$ are canonical guarded sums.\\
~\\
\textit{\textbf{Case of $\widehat{\Delta}$ transition.}} Since $(U,F,V)\in \mathcal{R}^{\prime}$ and $(P,E,Q)\in \mathcal{R}$, we have to prove that if $U\xrightarrow[\lambda]{\widehat{\Delta}}U^{\prime}$, then $V\xLongrightarrow[\rho,\rho_1,\rho^{\prime}]
{\widehat{\Delta}^c}V^{\prime}$ and $(U^{\prime},F^{\prime},V^{\prime})\in \mathcal{R}^{\prime}$.

Now, we assume that $S\oplus_C P\xrightarrow[\lambda]{\widehat{\Delta}}
S^{\prime}\oplus_{C^{\prime}}P^{\prime}$. Because we have considered the communications between $S$ and $P$ in the first part above. Here, we only consider the observable transitions from $S$ and $P$ without any communication. So, we have the following two transitions
$S\xrightarrow[\lambda_1]{\widehat{\Delta}_1}S^{\prime}$ and $P\xrightarrow[\lambda_2]{\widehat{\Delta}_2}P^{\prime}$ for $S$ and $P$, respectively, where $\widehat{\Delta}_1\uplus\widehat{\Delta}_2=\widehat{\Delta}$. For the residual functions, we have
$\lambda:|S^{\prime}|\cup |P^{\prime}|\rightarrow |S|\cup |P|$, $\lambda_1:|S^{\prime}|\rightarrow|S|$ and  $\lambda_2:|P^{\prime}|\rightarrow|P|$ with  $\lambda(s^{\prime})=\lambda_1(s^{\prime})$ for any $s^{\prime}\in S^{\prime}$ and $\lambda(p^{\prime})=\lambda_2(p^{\prime})$ for any $p^{\prime}\in P^{\prime}$.

Since $(P,E,Q)\in \mathcal{R}$  and $P\xrightarrow[\lambda_2]{\widehat{\Delta}_2}P^{\prime}$
we have $Q\xLongrightarrow[\rho,\rho_2,\rho^{\prime}]
{\widehat{\Delta}_2^c}Q^{\prime}$
for any $p:\alpha\cdot({\vec L})\in \widehat{\Delta}_2$ there exists $q:\alpha\cdot({\vec M})\in \widehat{\Delta}_2^{c}$ such that $(p,\rho(q))\in E$ and $(P^{\prime},E^{\prime},Q^{\prime})\in \mathcal{R}$ for some $E^{\prime}\subseteq|P^{\prime}|\times|Q^{\prime}|$
such that if $(p^{\prime},q^{\prime})\in E^{\prime}$ then $(\lambda_2(p^{\prime}),\rho\rho_2\rho^{\prime}(q^{\prime}))\in E$.

Therefore, we have $V\xLongrightarrow[\nu,\nu_1,\nu^{\prime}]
{\widehat{\Delta}^c}V^{\prime}$
where $\widehat{\Delta}^c = \widehat{\Delta}_1 \uplus \widehat{\Delta}_2^c$, $V^{\prime}=S^{\prime}\oplus_{D^{\prime}}Q^{\prime}$ with $D^{\prime}=\{(s^{\prime},q^{\prime})\in |S^{\prime}|\times|Q^{\prime}|\mid (\nu_1(s^{\prime}),\rho\nu_1\rho^{\prime}(q^{\prime}))\in D\}$.
$V\xLongrightarrow[\nu,\nu_1,\nu^{\prime}]
{\widehat{\Delta}^c}V^{\prime}$ can be decomposed as
$$S\oplus_D Q \xrightarrow[\nu]{\tau^{\ast}}S\oplus_{D_1}Q_1 \xrightarrow[\nu_1]{\widehat{\Delta}^c}
S^{\prime}\oplus_{D_1^{\prime}}Q_1^{\prime}
\xrightarrow[\nu^{\prime}]{\tau^{\ast}}
S^{\prime}\oplus_{D^{\prime}}Q^{\prime}$$
with $\nu=\mbox{Id}_{|S|}\cup \rho$ and $\nu^{\prime}=\mbox{Id}_{|S^{\prime}|}\cup \rho^{\prime}$. $\nu_1:|S^{\prime}|\cup |Q_1^{\prime}|\rightarrow |S|\cup |Q_1|$ with $\nu_1(s^{\prime})=\lambda_1(s^{\prime})$ for any $s^{\prime}\in |S^{\prime}|$ and $\nu_1(q_1^{\prime})= \rho_2(q_1^{\prime})$ for any $q_1^{\prime}\in |Q_1^{\prime}|$.

Let $F^{\prime}\subseteq |U^{\prime}|\times |V^{\prime}|$ be defined as $F^{\prime}= \mbox{Id}_{|S^{\prime}|}\cup E^{\prime}$.
For  $(u^{\prime},v^{\prime})\in F^{\prime}$, if $u^{\prime}\in |S^{\prime}|$ or $v^{\prime}\in |S^{\prime}|$, we must have $u^{\prime}=v^{\prime}$. If $u^{\prime}\notin |S^{\prime}|$ and $v^{\prime}\notin |S^{\prime}|$ then we have $(u^{\prime},v^{\prime})\in E^{\prime}$. Hence $(\lambda(u^{\prime}),\nu\nu_1\nu^{\prime}(v^{\prime}))=
(\lambda_2(u^{\prime}),\rho\rho_2\rho^{\prime}(v^{\prime}))\in E$.

Moreover, the triple $(C^{\prime},D^{\prime},E^{\prime})$ is adapted: let $(p^{\prime},q^{\prime})\in E^{\prime}$ and $s^{\prime}\in |S^{\prime}|$. We have $(\lambda_2(p^{\prime}),\rho\rho_2\rho^{\prime}(q^{\prime}))\in E$. We have $(s^{\prime},p^{\prime})\in C^{\prime}$ iff $(\lambda(s^{\prime}),\lambda(p^{\prime}))\in C$ iff $(\lambda_1(s^{\prime}),\lambda_2(p^{\prime}))\in C$ iff $(\lambda_1(s^{\prime}),\rho\rho_2\rho^{\prime}(q^{\prime}))\in D$ iff
$(\nu\nu_1\nu^{\prime}(s^{\prime}),\nu\nu_1\nu^{\prime}(q^{\prime}))\in D $ iff $(s^{\prime},q^{\prime})\in D^{\prime}$.
\qed\end{proof}
The proof for Theorem \ref{bisimilationCongruence}.
\begin{proof}
Let $\mathcal{R}$ be a localized early weak bisimulation. Let $R$ be a $Y${ -context}. We define a new localized relation denoted by $R[\mathcal{R}/Y]$:
\begin{itemize}
  \item if $Y = R$ then $R[\mathcal{R}/Y] = \mathcal{R}$
  \item if $Y\neq R$ then we make $(P^{\prime},E^{\prime},Q^{\prime})\in R[\mathcal{R}/Y]$ if there exist $(P,E,Q)\in \mathcal{R}$, $E^{\prime}= \mbox{Id}_{|R|}$, $P^{\prime}=R[P/Y]$ and $Q^{\prime}= R[Q/Y]$. Since $R\neq Y$, it is obvious that $|P^{\prime}|=|Q^{\prime}|= |R|$.
\end{itemize}
We define a localized relation $\mathcal{R}^+$ as the union of $\mathcal{I}$ (the set of all triples $(U,E,U)$ where $U\in {\sf Proc}$ and $E = \mbox{Id}_{|U|}$), the parallel extension $\mathcal{R}^{\prime}$ of $\mathcal{R}$ and all the relations of the shape $R[\mathcal{R}/Y]$ for all $Y$-context $R$. Then what we have to do is to prove that $\mathcal{R}^+$ is a localized early weak bisimulation. It is easy to check that $\mathcal{R}^+$ is symmetric.

Let $(U,F,V)\in \mathcal{R}^+$ and we have to analyse the two following situations:
\begin{itemize}
  \item[(1)] $U\xrightarrow[\mu]{\tau}U^{\prime}$
  \item[(2)] or $U\xrightarrow[\mu]{\widehat{\Delta}}U^{\prime}$
\end{itemize}
In each case, we analyse all the possible transitions from the challenger, and then we show that there are corresponding transitions of the defender to respond to the challenger. We consider all the possible relations from $\mathcal{R}^+$. We analyse the two cases in details.
\begin{itemize}
  \item For case (1) we must show that $V\xrightarrow[\nu]{\tau^{\ast}}V^{\prime}$ with $(U^{\prime},F^{\prime}, V^{\prime})\in \mathcal{R}^+$ for some $F^{\prime}\subseteq |U^{\prime}|\times|V^{\prime}|$ such that for any $(u^{\prime},v^{\prime})\in F^{\prime}$, we have $(\mu(u^{\prime}),\nu(v^{\prime}))\in F$.
  \item For case (2) we must show that $V\xLongrightarrow[\nu,\nu_1,\nu^{\prime}]
      {\widehat{\Delta}^c}V^{\prime}$ with $(U^{\prime},F^{\prime},V^{\prime})\in \mathcal{R}^+$ and for any pair of labels $p:\alpha\cdot({\vec L})\in \widehat{\Delta}$ and $q:\alpha\cdot({\vec M})\in \widehat{\Delta}^c$, $(p,\nu(q))\in F$.
      And for some $F^{\prime}\subseteq|U^{\prime}|\times|V^{\prime}|$ such that for any $(u^{\prime},v^{\prime})\in F^{\prime}$, we have $(\mu(u^{\prime}),\nu\nu_1\nu^{\prime}(v^{\prime}))\in F$.
\end{itemize}
Now, we analyse the possible relations in $\mathcal{R}^+$.

The case where $(U,F,V)\in \mathcal{I}$ is trivial.

If $(U,F,V)\in \mathcal{R}^{\prime}$, we can directly apply  Proposition \ref{propositionparallel} to both the cases.

Assume that $(U,F,V)\in R[\mathcal{R}/Y]$ for some $Y$-context $R$, so that $U=R[P/Y]$ and $V= R[Q/Y]$ with $(P,E,Q)\in \mathcal{R}$ such that $F=E$ if $R = Y$ and $F=\mbox{Id}_{|R|}$ otherwise. If $ R = Y$, we can directly use the fact that $\mathcal{R}$ is a localized weak bisimulation to show that $V^{\prime}$ and $F^{\prime}$ satisfy the required conditions.

At last we consider $R\neq Y$, so we have $F = \mbox{Id}_{|R|}$.
In this paper, we only focus on canonical processes.
For the canonical guarded sum ${\sf cs}(P)$, its prefixed form ${\sf cs}(P)$ is of the form $pre\cdot(Q_1,\ldots,Q_n)+T$,
where $pre$ is a prefix, $T$ is a canonical guarded sum and $Q_1,\ldots,Q_n$ are canonical processes.

By the definition of the $Y$-context, there is exactly one $r\in |R|$ such that $Y$ occurs free in $R(r)$. And ${\sf cs}(R(r))= f(x)\cdot{\vec R} + \widetilde{R}$ and $Y$ does not occur free in $\widetilde{R}$ and occurs exactly in one of the processes $R_1,\ldots,R_n$. Without loss of generality we assume that $R_1$ is a $Y$-context and $Y$ does not occur free in $R_2,\ldots,R_n$.

We assume that $R_1\neq Y$. In both cases (1) and (2), we have $U^{\prime}= R^{\prime}[P/Y]$ with $R\xrightarrow[\mu]{\tau}R^{\prime}$ (case (1)) or $R\xrightarrow[\mu]{\widehat{\Delta}}R^{\prime}$ (case (2)). Let $V^{\prime}=R^{\prime}[Q/Y]$.
In case (1), we have $V\xrightarrow[\mu]{\tau}V^{\prime}$ and in case (2) we have $V\xrightarrow[\mu]{\widehat{\Delta}^c}V^{\prime}$. Since $Y\neq R^{\prime}$, we have $(U^{\prime},\mbox{Id}_{|R^{\prime}|}, V^{\prime})\in \mathcal{R}^{+}$ for $(P,E,Q)\in \mathcal{R}$. The residual condition is obviously satisfied in both cases.

At last we assume that $R_1=Y$.

\textit{\textbf{For case (1)}}. There are two cases to consider the locations $s,t\in |U|$ involved in the transition $U\xrightarrow[\mu]{\tau}U^{\prime}$. The case $s\neq r$ and $t\neq r$ is similar to the case above where $R_1 \neq Y$. The other two cases are  the case $s=r$ (hence $t\neq r$) and the symmetric case $t=r$ (hence $s\neq r$). We just consider the case $s=r$.

So $U(t)=R(t)=\overline{f}(e)\cdot{\vec T}+\widetilde{T}$ with ${\sf eval}(e)=v$ and the guarded sum $R(r)$ has an unique summand involved in the transition $U\xrightarrow[\mu]{\tau}U^{\prime}$ and this summand is of the form $f(x)\cdot{\vec S}$ (called active summand in the text that follows).

If the active summand is $f(x)\cdot{\vec R}$ then we have $U(r)= f(x)\cdot(P,R_2,\ldots,R_n) + \widetilde{S}$. $U^{\prime}$ can be written as $U^{\prime}=R^{\prime}\oplus_C P\{v/x\}$ for some process $R^{\prime}$ which can be defined using only $R$ and $C\subseteq |R^{\prime}|\times |P\{v/x\}|$. $R^{\prime}$ is defined as follows:
\begin{itemize}
  \item $|R^{\prime}|= (|R|\setminus\{t,r\})\cup \bigcup_{i=2}^n |R_i\{v/x\}|\cup \bigcup_{i=1}^n |T_i|$
  \item and $\frown_{R^{\prime}}$ is the least symmetric relation on $|R^{\prime}|$ such that $r^{\prime}\frown_{R^{\prime}} t^{\prime}$
      if $r^{\prime}\frown_{R_i\{v/x\}}t^{\prime}$ for some $i\in \{2,\ldots,n\}$,
      or $r^{\prime}\frown_{T_i}t^{\prime}$ for some $i\in\{1,\ldots,n\}$,
      or $(r^{\prime},t^{\prime})\in (\bigcup_{i=2}^n |R_i\{v/x\}|)\times (\bigcup_{i=1}^n |T_i|)$,
      or $\mu(r^{\prime})\frown_R\mu(t^{\prime})$ with $r^{\prime}\notin \bigcup_{i=2}^n |R_i\{v/x\}|$ or $t^{\prime}\notin \bigcup_{i=1}^n|T_i|$. 

\end{itemize}
where the residual function $\mu:|U^{\prime}|\rightarrow |U|$ is given by $\mu(r^{\prime})= r$ if $r^{\prime}\in |P\{v/x\}|\cup\bigcup_{i=2}^n|R_i\{v/x\}|$,
$\mu(r^{\prime})= t$ if $r^{\prime}\in \bigcup_{i=1}^n|T_i|$, and $\mu(r^{\prime})=r^{\prime}$ otherwise.

The relation $C$ is defined as follows: given $(r^{\prime},p)\in |R^{\prime}|\times |P\{v/x\}|$,
one has $(r^{\prime},p)\in C$ if $r^{\prime}\in |T_1|$, or $r^{\prime}\notin \bigcup_{i=2}^n|R_i\{v/x\}|\cup \bigcup_{i=1}^n|T_i|$ and $r^{\prime}\frown_R r$.

Let $V^{\prime} = R^{\prime}\oplus_D Q\{v/x\}$, where $ D\subseteq |R^{\prime}|\times |Q\{v/x\}|$ is defined similarly in the way for $C$ by replacing $P\{v/x\}$ by $Q\{v/x\}$. From $(p,q)\in E$ and the definitions of $C$ and $D$, we have $(r^{\prime},p)\in C$ iff $(r^{\prime},q)\in D$. So $(C,D,E)$ is adapted. We can make the reduction on $V$, such that $V\xrightarrow[\nu]{\tau}V^{\prime}$
for the residual function $\nu$ which is defined like $\mu$ by replacing $P\{v/x\}$ by $Q\{v/x\}$. We have $(U^{\prime},F^{\prime},V^{\prime})\in \mathcal{R}^{\prime}\subseteq \mathcal{R}^{+}$ where $F^{\prime}= \mbox{Id}_{|R^{\prime}|}\cup E$. If $(u^{\prime},v^{\prime})\in F^{\prime}$, then we have $\mu(u^{\prime})=\nu(v^{\prime})$, that is $(\mu(u^{\prime}),\mu(v^{\prime}))\in F$ so that the condition on residuals holds.

If the active summand is not $f(x)\cdot{\vec R}$, then we have $V\xrightarrow[\mu]{\tau}U^{\prime}$ (both $P$ and $Q$ are discarded in the corresponding reductions, respectively). We just finish the proof because of $(U^{\prime},\mbox{Id}_{|U^{\prime}|},U^{\prime})\in \mathcal{I}\subseteq\mathcal{R}^{\prime}$.

\textit{\textbf{For case (2)}}. In the transition $U\xrightarrow[\mu]{\widehat{\Delta}}U^{\prime}$, if $r$ is not mentioned in $\widehat{\Delta}$, then we have $R[P/Y]=U\xrightarrow[\mu]{\widehat{\Delta}}
U^{\prime}=R^{\prime}[P/Y]$. We also have $R[Q/Y]=V\xrightarrow[\mu]{\widehat{\Delta}}V^{\prime}=R^{\prime}[Q/Y]$ so $(U^{\prime},\mbox{Id}_{|R^{\prime}|},V^{\prime})\in R^{\prime}[\mathcal{R}/Y]\subseteq \mathcal{R}^{+}$ and the residual condition is satisfied.

If $r:\alpha\cdot(\vec{L})$ is mentioned in $\widehat{\Delta}$, then there exists exactly one of the summands of the guarded sum $R(r)$ being the prefixed process preforming action $\alpha$ in $U\xrightarrow[\mu]{\widehat{\Delta}}U^{\prime}$.

The case where the active summand is not $f(x)\cdot(P,R_2,\ldots,R_n)$ is similar to the previous case, because $P$ is discarded in the transition.

For $R[P/Y]$, we can rewrite it as $R_1\oplus_{C_1} R(r)$, where $R_1(s)=R(s)$ for $s\in (|R|\setminus\{r\})$, and $(s,r)\in C_1$ if $s\frown_R r$.

If the active summand is $f(x)\cdot(P,R_2,\ldots, R_n)$, then
$U=R[P/Y]= R_1\oplus_{C_1} R(r) \xrightarrow[\mu]{\widehat{\Delta}}U^{\prime}=R_1^{\prime}\oplus_{C_1^{\prime}}
(P\{v/x\}\oplus R_2\{v/x\}\cdots\oplus R_n\{v/x\})$ for $v\in {\bf Val}$.

We rewrite $U^{\prime}$ as $R^{\prime}\oplus_C P\{v/x\}$ for $v\in {\bf Val}$, where $R^{\prime}$ is defined by
\begin{itemize}
  \item $|R^{\prime}|= |R_1^{\prime}|\cup \bigcup_{i=2}^n|R_i\{v/x\}|$
       and $\frown_{R^{\prime}}$ is the least symmetric relation on $|R^{\prime}|$ such that $r^{\prime}\frown_{R^{\prime}}t^{\prime}$ if $r^{\prime}\frown_{R_i\{v/x\}}t^{\prime}$ for some $i\in \{2,\ldots, n\}$ or $\mu(r^{\prime})\frown_R \mu(t^{\prime})$.
\end{itemize}

The relation $C\subseteq |R^{\prime}|\times |P\{v/x\}|$ is defined by $(r^{\prime},q)\in C$ if $r^{\prime}\notin \bigcup_{i=2}^n|R_i|$ and $\mu(r^{\prime})\frown_R r$.

Then we have
$V=R[Q/Y]= R_1\oplus_{C_1} R(r) \xrightarrow[\mu]{\widehat{\Delta}^c}V^{\prime}=R_1^{\prime}\oplus_{C_1^{\prime\prime}}
(Q\{v/x\}\oplus R_2\{v/x\}\ldots\oplus R_n\{v/x\})$ for $v\in {\bf Val}$.

We rewrite $V^{\prime}$ as $R^{\prime}\oplus_D Q\{v/x\}$ where $R^{\prime}$ is defined as above and  $D$ is defined like $C$ by replacing $P\{v/x\}$ by $Q\{v/x\}$.
Then we have $(U^{\prime},F^{\prime},V^{\prime})\in \mathcal{R}^{\prime}\subseteq\mathcal{R}^{+}$ where $F^{\prime}=\mbox{Id}_{|R^{\prime}|}\cup E$
since $(C,D,E)$ is adapted. Moreover the condition on residuals is obviously satisfied.

The symmetric case that ${\sf cs}(R(r))= \overline{f}(e)\cdot{\vec R} + \widetilde{R}$ with ${\sf eval}(e)=v$ and $Y$ does not occur free in $\widetilde{R}$ and occurs exactly in one of the processes $R_1,\ldots,R_n$, is similar. So, we show the fact that $\mathcal{R}^{+}$ is a localized early weak bisimulation.

We can now prove that $\approx$ is a congruence. Assume that $P\approx Q$ and let $R$ be a $Y$-context. Let $E\subseteq|P|\times|Q|$ and let $\mathcal{R}$ be a localized early weak bisimulation such that $(P,E,Q)\in \mathcal{R}$.
Then we have $(R[P/Y],\mbox{Id}_{|R|}, R[Q/Y])\in R[\mathcal{R}/Y]\subseteq \mathcal{R}^{+}$ and hence $R[P/Y]\approx R[Q/Y]$ since $\mathcal{R}^{+}$ is a localized early weak bisimulation.
\qed\end{proof}

\subsection{Proof for Completeness}


Given a process $P\in{\bf Pr}$, we say that ${\sf Sort}(P)\subseteq \Sigma$ is the sort of $P$. ${\sf Sort}$ is the least function, extracting symbols from processes, such that: ${\sf Sort}(X)= {\sf Sort}(\ast)= {\sf Sort}({\bf 0})=\emptyset$, 
${\sf Sort}(P\backslash I)={\sf Sort}(P)\setminus I $, ${\sf Sort}(G\langle\Phi\rangle)=\bigcup_{p\in |G|} {\sf Sort}(\Phi(p))$, ${\sf Sort}({\bf if}~b~{\bf then}~P~{\bf else}~Q)={\sf Sort}(P)\cup {\sf Sort}(Q)$, ${\sf Sort}(P+Q)={\sf Sort}(P)\cup {\sf Sort}(Q)$, ${\sf Sort}(f(x)\cdot(P_1,\ldots,P_n))= {\sf Sort}(\overline{f}(e)\cdot(P_1,\ldots,P_n)) = \{f\}\cup\bigcup_{i=1}^n{\sf Sort}(P_i)$, and ${\sf Sort}(P)\subseteq {\sf Sort}(A({\vec v}))$ with $A({\vec x})\stackrel{\rm def}{=}P$.

In the proof below, we write $\sum_{i\in I} P_i$ to mean the sum of all $P_i$, for $i\in I$. In a statement, we say that a co-symbol $f$ is fresh if $\overline{f}$ does not occur in the sort of the processes in the statement, and similarly for fresh symbols.

\begin{lemma}\label{bisim_barb_star}
For any process $P$, $P\oplus_C \ast \approx P$ and $P\oplus_C \ast \stackrel{\bullet}{\approx} P $ for any $C\subseteq |P|\times |\ast|$.
\end{lemma}
\begin{proof}
Let $|P\oplus_C \ast| = |P|\cup\{l\}$ for some $l\notin |P|$. And we can build the localized relation $\mathcal{R}=\{(P\oplus_C \ast, \mathrm{Id}_{|P|}, P), (P,\mathrm{Id}_{|P|},P\oplus_C \ast)\mid C\subseteq |P|\times |\ast|\}$. It is easy to show $\mathcal{R}$ is a localized early weak bisimulation. It is similar to show $P\oplus_C \ast \stackrel{\bullet}{\approx} P$.
\qed\end{proof}
Proof for Lemma \ref{two_bisim}.
\begin{proof}
The inclusion $\approx\subseteq \approx_{\omega}$ is easy. One proves that $\approx\subseteq \approx_n$ for all $n$, using the fact that $\approx$ is a weak bisimulation.

For the converse, we show that the set
$$\mathcal{R}\stackrel{\rm def}{=}\{(P,E,Q)\mid (P,E,Q)\in \approx_{\omega}, E\subseteq|P|\times|Q|\}$$
is a weak bisimulation. Take $(P,E,Q)\in \mathcal{R}$ and $E\subseteq|P|\times|Q|$, and suppose $P\xrightarrow[\lambda]{\tau}P^{\prime}$. We need a matching transition from $Q$. For all $n$, as
$(P,E,Q)\in\approx_{n+1}$, there is $Q_n$ such that $Q\xrightarrow[\rho]{\tau^{\ast}}Q_n$, $(P^{\prime},E^{\prime},Q_n)\in\approx_n$ and $E^{\prime}\subseteq|P^{\prime}|\times|Q_n|$. From the definitions of $\approx$ and $\approx_{\omega}$, it is easy to check that the residual conditions are satisfied. Because the LLTS is finitely-branching, the set $\{Q_i\mid Q\xrightarrow[\rho]{\tau^{\ast}}Q_i\}$ is finite. Thus, there is at least a $Q_i$ such that $(P^{\prime},E^{\prime},Q_i)\in\approx_n$ with $E^{\prime}\subseteq |P^{\prime}|\times|Q_i|$ holds for infinitely many $n$. As the relations $\{\approx_n\}_n$ are decreasing by definition, $(P^{\prime},E^{\prime},Q_i)\in\approx_n$ holds for all $n$. Hence $(P^{\prime},E^{\prime},Q_i)\in\approx_{\omega}$ and $(P^{\prime},E,Q_i)\in\mathcal{R}$.

The case $P\xlongrightarrow[\lambda]{\widehat{\Delta}}P^{\prime}$ is similar.
\qed\end{proof}
Proof for Theorem \ref{completeness}.
\begin{proof}
{\it We have to pay more attentions to the $n$-ary symbols and non-sequential semantics.}
We need to consider two cases $P\xrightarrow[\lambda]{\tau}P^{\prime}$ and $P\xrightarrow[\lambda]{\Delta}P^{\prime}$ by induction on $n$. The cases for $Q$ are symmetric.

{\bf (Single input)} We first consider the case when the single-labelled transition is an input.
When $n=0$ there is nothing to prove. Suppose $n>0$.
Then there exist $p:f_i i\cdot({\vec L})$, $\lambda$ and $P^{\prime}$ such that $P\xrightarrow[\lambda]{p:f_i i\cdot({\vec L})}P^{\prime}$, but $(P^{\prime},|P^{\prime}|\times|Q^{\prime}|,Q^{\prime})\notin \approx_{n-1}$ for all $Q^{\prime}$ such that $Q\xLongrightarrow[\sigma,\sigma_1,\sigma^{\prime}]{q:f_i i\cdot({\vec H})}Q^{\prime}$. Since the LLTS is image-finite, $\{Q^{\prime}\mid Q\xLongrightarrow[\sigma,\sigma_1,\sigma^{\prime}]{q:f_i i\cdot({\vec H})}Q^{\prime}\}=\{Q_j\mid j\in J\}$ for some finite set $J$.
Since we use $(P^{\prime},|P^{\prime}|\times|Q^{\prime}|,Q^{\prime})$, the residual conditions are obviously satisfied.
Appealing to the induction hypothesis, for each $j\in J$,
$$P\mid (M + \overline{g}(0)\cdot(\ast))\mid D  ~ /\!\!\!\!\!\!\stackrel{\bullet}{\approx} Q^{\prime} \mid (M + \overline{g}(0)\cdot(\ast))\mid D. \eqno{\rm (IH1)}$$
(Here we let $R$ be of the form $(M + \overline{g}(0)\cdot(\ast))\mid D$.)

Let $\overline{d}$, $\overline{c^{\prime}}$ and $\overline{c_j}$ ($j\in J $) be fresh co-symbols, and set
$$D\stackrel{\rm def}{=}\overline{d}(0)\cdot(D),$$
$$M\stackrel{\rm def}{=}\overline{f_i}(i)\cdot(N,\ast,\ldots,\ast), \hbox{ and }$$
$$N\stackrel{\rm def}{=}\overline{c^{\prime}}(0)\cdot(\ast)+\sum_{j\in J}d(x)\cdot(M_j + \overline{c_j}(0)\cdot(\ast)).$$
Because we focus on canonical processes, we use $D$ to interact with $M + \overline{g}(0)\cdot(\ast)$ to generate internal reductions.
We show that $(M + \overline{g}(0)\cdot(\ast))\mid D$ is as required by $R$.
So suppose that $g$ is fresh.
Let $Q^{\prime}$ be any process such that $Q\xrightarrow[\rho_0]{\tau^{\ast}}Q^{\prime}$ (i.e. $Q\xrightarrow[]{}^{\ast} Q^{\prime}$).
Let $A\stackrel{\rm def}{=}P\mid (M+\overline{g}(0)\cdot(\ast))\mid D$ and $B\stackrel{\rm def}{=} Q^{\prime}\mid (M+\overline{g}(0)\cdot(\ast))\mid D$, and suppose, for a contradiction, such that $A\stackrel{\bullet}{\approx}B$. We have
$$A\xrightarrow[]{}A^{\prime}\stackrel{\rm def}{=}P^{\prime}\mid (N\oplus \ast\oplus\cdots\oplus\ast)\mid D.$$

Since $A\stackrel{\bullet}{\approx}B$, there is $B^{\prime}$ such that $B\xrightarrow[]{}^{\ast}B^{\prime}\stackrel{\bullet}{\approx} A^{\prime}$. Since $A^{\prime}\downarrow_{\{\overline{g}\}}$ does not hold, $B^{\prime}\downarrow_{\{\overline{g}\}}$ should not hold either. The only way this is possible is if $J\neq \emptyset$ and
$$B^{\prime}\stackrel{\rm def}{=}Q_j\mid (N\oplus \ast\oplus\cdots\oplus\ast)\mid D$$
for some $j\in J$ and $Q^{\prime}\xrightarrow[\rho_1]{q:f_i i\cdot({\vec H})}Q_j$. We now exploit the inductive hypothesis on $P^{\prime}$, $Q_j$ and $M_j$. We have
$$A^{\prime}\xrightarrow[]{}A^{\prime\prime}_j\stackrel{\rm def}{=}P^{\prime}\mid ((M_j + \overline{c_j}(0)\cdot(\ast))\oplus \ast \cdots\oplus \ast)\mid D$$
through internal reduction between $N$ and $D$.

Since $B^{\prime}\stackrel{\bullet}{\approx} A^{\prime}$, there is $B_j^{\prime\prime}$ such that $B^{\prime}\xrightarrow[]{}^{\ast}B_j^{\prime\prime}\stackrel{\bullet}{\approx}A^{\prime\prime}_j$. Without loss of generality, since $A^{\prime\prime}_j\downarrow_{\{\overline{c_j}\}}$ we must have $B^{\prime\prime}_j\downarrow_{\{\overline{c_j}\}}$. The only possibility is
$$B^{\prime\prime}_j\stackrel{\rm def}{=}Q_j^{\prime}\mid ((M_j + \overline{c_j}(0)\cdot(\ast))\oplus\ast\cdots\oplus\ast)\mid D$$
for some $Q_j^{\prime}$ such that $Q_j\xrightarrow[\rho_2]{\tau^{\ast}}Q^{\prime}_j$. Thus we have $Q\xLongrightarrow[\rho_0,\rho_1,\rho_2]{q:f_i i\cdot({\vec H})}Q^{\prime}_j$.

From Lemma \ref{bisim_barb_star}, we have
$$A^{\prime\prime}_j\stackrel{\bullet}{\approx} P^{\prime}\mid (M_j + \overline{c_j}(0)\cdot(\ast))\mid D$$
and
$$B^{\prime\prime}_j\stackrel{\bullet}{\approx} Q_j^{\prime}\mid (M_j + \overline{c_j}(0)\cdot(\ast))\mid D.$$
Since $\stackrel{\bullet}{\approx}$ is an equivalence relation, we have
$$P^{\prime}\mid (M_j + \overline{c_j}(0)\cdot(\ast))\mid D \stackrel{\bullet}{\approx}  Q_j^{\prime}\mid (M_j + \overline{c_j}(0)\cdot(\ast))\mid D.$$
By induction hypothesis ${\rm (IH1)}$, it is a contradiction. Hence $B~/\!\!\!\!\!\!\stackrel{\bullet}{\approx}A$ as required.

{\bf (Single output)} The case for single-labelled output transition $P\xrightarrow[\lambda]{p:\overline{f_i} i\cdot({\vec L})}P^{\prime}$ is similar.

{\bf (Single $\tau$-transition)} For the case $\tau$-transition $P\xrightarrow[\lambda]{\tau}P^{\prime}$,
let $\overline{d}$, $\overline{g}$ and $\overline{c_j}$ ($j\in J $) be fresh co-symbols, and
we set
$$D\stackrel{\rm def}{=}\overline{d}(0)\cdot(D),$$
$$M\stackrel{\rm def}{=}\sum_{j\in J}d(x)\cdot(M_j + \overline{c_j}(0)\cdot(\ast)),{\rm and}$$
$$R\stackrel{\rm def}{=} (M + \overline{g}(0)\cdot(\ast))\mid D.$$
Then the proof is similar.

{\bf (Multi-labelled transition)}
For $n=0$ there is nothing to prove. Suppose $n>0$.
Then there exist $\lambda$, $\widehat{\Delta}$ and $P^{\prime}$ such that $P\xrightarrow[\lambda]{\widehat{\Delta}}P^{\prime}$, but  $(P^{\prime},|P^{\prime}|\times|Q^{\prime}|,Q^{\prime})\notin\approx_{n-1}$ for all $Q^{\prime}$ such that $Q\xLongrightarrow[\sigma,\sigma_1,\sigma^{\prime}]{\widehat {\Delta}^c}Q^{\prime}$. Since the LLTS is image-finite,
$\{Q^{\prime}\mid Q\xLongrightarrow[\sigma,\sigma_1,\sigma^{\prime}]{\widehat {\Delta}^c}Q^{\prime}\}=\{Q_j\mid j\in J\}$.
Since we use $(P^{\prime},|P^{\prime}|\times|Q^{\prime}|,Q^{\prime})$, the residual conditions are obviously satisfied.
Appealing to the induction hypothesis, for each $j\in J$,
$$P\mid R   ~ /\!\!\!\!\!\!\stackrel{\bullet}{\approx} Q^{\prime} \mid R . \eqno{\rm (IH2)}$$
(Here we let $R$ be of the form $((M_1+\overline{g_1}(0)\cdot(\ast))\oplus\cdots\oplus (M_k+\overline{g_k}(0)\cdot(\ast)))\mid D$.)

Let $k={\sf size}(\widehat{\Delta})$.
Let $\overline{d}$, $\overline{c_i^{\prime}}$ and $\overline{c_{ij}}$ ($j\in J $ and $i\in\{1,\ldots,k\}$) be fresh co-symbols, and set
$$D\stackrel{\rm def}{=}\overline{d}(0)\cdot(D),$$
$$M_i\stackrel{\rm def}{=}\overline{f_i}(i)\cdot(N_i,\ast,\ldots,\ast), \hbox{ and }$$
$$N_i\stackrel{\rm def}{=}\overline{c_i^{\prime}}(0)\cdot(\ast)+\sum_{j\in J}d(x)\cdot(M_{ij} + \overline{c_{ij}}(0)\cdot(\ast)).$$
Set $R\stackrel{\rm def}{=}((M_1+\overline{g_1}(0)\cdot(\ast))\oplus\cdots\oplus (M_k+\overline{g_k}(0)\cdot(\ast)))\mid D$, and $g_i$ is fresh for $i\in\{1,\ldots,k\}$. Because we focus on canonical processes, we use $D$ to interact processes $(M_i+\overline{g_i}(0)\cdot(\ast))$, $i\in\{1,\ldots,k\}$, to generate internal reductions.

Let $Q^{\prime}$ be any process such that $Q\xrightarrow[\rho_0]{\tau^{\ast}}Q^{\prime}$ (i.e. $Q\xrightarrow[]{}^{\ast} Q^{\prime}$).
Let $A\stackrel{\rm def}{=}P\mid R\mid D$ and $B\stackrel{\rm def}{=} Q^{\prime}\mid R \mid D$, and suppose, for a contradiction, such that $A\stackrel{\bullet}{\approx}B$.
From $P\xrightarrow[\lambda]{{\widehat \Delta}}P^{\prime}$ and ${\widehat \Delta}$ is pairwise unrelated, through Diamond Property (Lemma \ref{diomand}), we have
$$A\xrightarrow[]{}^{\ast}A^{\prime}\stackrel{\rm def}{=}P^{\prime}\mid ((N_1\oplus \ast\oplus\cdots\oplus\ast)\oplus\cdots\oplus(N_k\oplus \ast\oplus\cdots\oplus\ast))\mid D$$
and $A^{\prime}\downarrow_{W_1}$ with $W_1=\{\overline{c_1^{\prime}},\ldots, \overline{c_k^{\prime}}\}$, but not $A^{\prime}\downarrow_{\{\overline{g_i}\}}$ ($i\in\{1,\ldots,k\}$). Since $A\stackrel{\bullet}{\approx}B$, there is $B^{\prime}$ such that $B\xrightarrow[]{}^{\ast}B^{\prime}$ and $A^{\prime}\stackrel{\bullet}{\approx}B^{\prime}$. It must be that $B^{\prime}\downarrow_{W_1}$ but not $B^{\prime}\downarrow_{\{\overline{g_i}\}}$ ($i\in\{1,\ldots,k\}$). The only way this is possible if $J\neq \emptyset$ and
$$B^{\prime}\stackrel{\rm def}{=}Q_j\mid ((N_1\oplus \ast\cdots\oplus\ast)\oplus\cdots\oplus(N_k\oplus \ast\cdots\oplus\ast))\mid D$$
for some $j\in J$ and $Q^{\prime}\xrightarrow[\rho_1]{\widehat{\Delta}^c}Q_j$ through Diamond Property.
We have $A^{\prime}\xrightarrow[]{}^{\ast}A_j^{\prime\prime}$ through internal reductions between $N_i$ ($i\in\{1,\ldots,k\}$) and $D$, and
$$A_j^{\prime\prime}\stackrel{\rm def}{=}P^{\prime}\mid (((M_{1j} + \overline{c_{1j}}(0)\cdot(\ast))\oplus \ast\cdots\oplus\ast)\oplus\cdots\oplus(((M_{kj} + \overline{c_{kj}}(0)\cdot(\ast)))\oplus \ast\cdots\oplus\ast))\mid D$$
and $A_j^{\prime\prime}\downarrow_{W_2}$ with $W_2=\{\overline{c_{1j}},\ldots, \overline{c_{kj}}\}$, but not $A_j^{\prime\prime}\downarrow_{\{\overline{c_i^{\prime}}\}}$ ($i\in\{1,\ldots,k\}$).

Since $B^{\prime}\stackrel{\bullet}{\approx} A^{\prime}$, there is $B_j^{\prime\prime}$ such that $B^{\prime}\xrightarrow[]{}^{\ast}B_j^{\prime\prime}\stackrel{\bullet}{\approx}A^{\prime\prime}_j$. Without loss of generality, since $A^{\prime\prime}_j\downarrow_{W_2}$ we must have $B^{\prime\prime}_j\downarrow_{W_2}$. The only possibility is
$$B^{\prime\prime}_j\stackrel{\rm def}{=}Q_j^{\prime}\mid (((M_{1j} + \overline{c_{1j}}(0)\cdot(\ast))\oplus \ast\cdots\oplus\ast)\oplus\cdots\oplus(((M_{kj} + \overline{c_{kj}}(0)\cdot(\ast)))\oplus \ast\cdots\oplus\ast)) \mid D$$
for some $Q_j^{\prime}$ such that $Q_j\xrightarrow[\rho_2]{\tau^{\ast}}Q^{\prime}_j$. Thus we have $Q\xLongrightarrow[\rho_0,\rho_1,\rho_2]{\widehat{\Delta}^c}Q^{\prime}_j$.

From Lemma \ref{bisim_barb_star}, we have
$$A^{\prime\prime}_j\stackrel{\bullet}{\approx} P^{\prime}\mid ((M_{1j} + \overline{c_{1j}}(0)\cdot(\ast))\oplus\cdots\oplus((M_{kj} + \overline{c_{kj}}(0)\cdot(\ast))))\mid D$$
and
$$B^{\prime\prime}_j\stackrel{\bullet}{\approx} Q_j^{\prime}\mid ((M_{1j} + \overline{c_{1j}}(0)\cdot(\ast))\oplus\cdots\oplus((M_{kj} + \overline{c_{kj}}(0)\cdot(\ast))))\mid D.$$
Since $\stackrel{\bullet}{\approx}$ is an equivalence relation, we have
$$P^{\prime}\mid ((M_{1j} + \overline{c_{1j}}(0)\cdot(\ast))\oplus\cdots\oplus((M_{kj} + \overline{c_{kj}}(0)\cdot(\ast))))\mid D\stackrel{\bullet}{\approx} $$
$$ Q_j^{\prime}\mid ((M_{1j} + \overline{c_{1j}}(0)\cdot(\ast))\oplus\cdots\oplus((M_{kj} + \overline{c_{kj}}(0)\cdot(\ast))))\mid D.$$
By induction hypothesis ${\rm (IH2)}$, it is a contradiction. Hence $B~/\!\!\!\!\!\!\stackrel{\bullet}{\approx}A$ as required.
\qed\end{proof}

\subsection{Proofs in Section \ref{Discussion}}
Proof for Proposition \ref{tree_automata}.
\begin{proof}
If $t=f(x)\cdot(t_1,\ldots,t_n)$ and $(Q,f(x),(Q_1,\ldots,Q_n))\in \mathcal{T}$, then one has $\langle\mathcal{A}\rangle_Q \mid {\sf proc}(t)\xrightarrow[]{} (\langle\mathcal{A}\rangle_{Q_1}\oplus \cdots\oplus \langle\mathcal{A}\rangle_{Q_n})\mid ({\sf proc}(t_1)\oplus\cdots\oplus {\sf proc}(t_n))$. Therefore, we can make a choice for the communications happening in the next step by $(\langle\mathcal{A}\rangle_{Q_1}\mid {\sf proc}(t_1))\oplus\cdots\oplus(\langle\mathcal{A}\rangle_{Q_n}\mid {\sf proc}(t_n))$. Since $t$ is recognized by $\mathcal{A}$ at state $Q$, this kind of choice is always possible at each step. And we can inductively check that $\langle\mathcal{A}\rangle_{Q_i}\mid {\sf proc}(t_i)$ can reduce to an idle process.
\qed\end{proof}

\subsection{Proofs in Section \ref{Exmaples}}
Proof for Proposition \ref{ABP_prop}.
\begin{proof}
Since $Q$ does not affect the reductions of $A\mid P_1$ and $P_2$, we only consider the reductions of $A\mid P_1$ and $P_2$. There are two cases for the reduction
\begin{itemize}
  \item either $(A\mid P_1(t,b))\mid P_2([],b)\xrightarrow[]{}^{\ast} (A\mid P_1(t_1,b))\mid P_2(t_2,b)$ with the concatenation of $t_1$ and $t_2$ equals to $t$, and this is a some stage of the reduction,
  \item or $(A\mid P_1(t,b))\mid P_2([],b)\xrightarrow[]{}^{\ast}(A\mid {\bf 0})\mid Succ(t)$, and this is the final successful stage of the reduction.
\end{itemize}

Induction on the length of $t_1$, we only show some cases and other cases are similar:
\begin{itemize}
  \item If ${\sf null}(t_1)$ is satisfied, then a possible reduction sequence is: sending the $End$ message, receiving acknowledge from $P_2$, passing the conditional evaluation in the sender and in the receiver respectively, then reducing to the final successful stage.
  \item If ${\sf null}(t_1)$ is not satisfied, then a possible reduction sequence is: sending the head of $t_1$, passing the conditional evaluation of the receiver, receiving acknowledge from the receiver, then reducing to the next stage of the reduction.
\end{itemize}
\qed\end{proof}
\end{document}